\documentclass[11pt, oneside]{article}   	% use "amsart" instead of "article" for AMSLaTeX format
\usepackage{geometry}                		% See geometry.pdf to learn the layout options. There are lots.
\usepackage{graphicx}				% Use pdf, png, jpg, or eps§ with pdflatex; use eps in DVI mode
\usepackage{tikz}

\usepackage{amssymb}
\usepackage{amsmath}
\usepackage{amsthm}
\usepackage{braket}
\usepackage{thmtools,thm-restate}

\usepackage[pdftex,colorlinks=true,linkcolor=blue,citecolor=blue,urlcolor=black]{hyperref}

\usepackage[linesnumbered,ruled]{algorithm2e}
\usepackage{authblk}

\title{Time and Space Efficient Quantum Algorithms for Detecting Cycles and Testing Bipartiteness}
\author[1]{Chris Cade\footnote{\href{mailto:me@example.com}{\nolinkurl{chris.cade@bristol.ac.uk}}}}
\author[1]{Ashley Montanaro}
\author[2]{Aleksandrs Belovs}
\affil[1]{School of Mathematics, University of Bristol, UK}
\affil[2]{Faculty of Computing, University of Latvia, Latvia}

\begin{document}
\newtheorem{theorem}{Theorem}
\newtheorem{lemma}{Lemma}
\newtheorem{definition}{Definition}
\newtheorem{claim}{Claim}

\maketitle
\begin{abstract}
We study space and time efficient quantum algorithms for two graph problems -- deciding whether an $n$-vertex graph is a forest, and whether it is bipartite. Via a reduction to the s-t connectivity problem, we describe quantum algorithms for deciding both properties in $\tilde{O}(n^{3/2})$ time and using $O(\log n)$ classical and quantum bits of storage in the adjacency matrix model. We then present quantum algorithms for deciding the two properties in the adjacency array model, which run in time $\tilde{O}(n\sqrt{d_m})$ and also require $O(\log n)$ space, where $d_m$ is the maximum degree of any vertex in the input graph. 
\end{abstract}

\section{Introduction}
Graph-theoretic problems are an important class of problems for which quantum algorithms can be shown to be faster than any possible classical algorithm. Examples of such problems include deciding whether there is a path between two vertices in a graph, or whether a graph is planar. The latter is exemplary of a subclass of graph problems that involve deciding whether or not a graph has a given property, which (besides planarity) includes properties such as containing a triangle, being bipartite, or being a forest. Many such graph properties are known to have efficient quantum algorithms \cite{durr2004quantum, childs_quantum_2012, ambainis2008quantum}. 
\\
\\
However, most of these algorithms use $\Omega(n\log n)$ bits of storage. Two exceptions are the work of Belovs and Reichardt \cite{belovs_span_2012}, who improve on the connectivity algorithm of \cite{durr2004quantum} by providing a time efficient ($\tilde{O}(n^{3/2})$) span-program-based algorithm for s-t connectivity that requires only logarithmic space, as well as \={A}ri\c{n}\v{s} \cite{arins_span-program-based_2015}, who describes a query-efficient ($O(n^{3/2})$) and space efficient ($O(\log n)$), but not time efficient, algorithm for testing bipartiteness. Other than this, however, little is known about the quantum space requirements for graph problems. It is desirable to have space efficient (i.e. $O(\log n)$) as well as time efficient algorithms for solving graph problems, since the graphs that will be good candidates for quantum algorithms are likely to be extremely large, and possibly given implicitly. Therefore reducing the amount of space (both classical and quantum bits) required to process them is important; moreover, it will be interesting to know how the space requirements of quantum algorithms relate to those of their classical counterparts. 
\\
\\
We study space-efficient algorithms for graph problems. In particular, we focus on the property of being a forest -- that is, deciding whether the input graph contains a cycle. We also consider the property of bipartiteness, which is characterised by containing no odd length cycle as a subgraph. The property of being a forest is minor-closed with a single forbidden minor (a triangle), whereas the property of bipartiteness is only subgraph-closed. Equivalently, both properties can be characterised by an infinite number of forbidden subgraphs (all cycles and all odd-length cycles, respectively). 

In this paper we consider two models for the input of a graph $G$ with vertex set $V$ and edge set $E$:
\begin{itemize}
\item {\bf The adjacency matrix model -- }The input is given as the adjacency matrix $A \in \{0,1\}^{n\times n}$, where $A_{ij}$ is 1 if and only if $(i, j) \in E$.
\item {\bf The adjacency array model -- }We are given the degrees of the vertices $d_1, d_2, ..., d_n$ and for every vertex $i$ an array with its neighbours $f_i : [d_i] \rightarrow [n]$, so that $f_i(j)$ returns the $j^{\text{th}}$ neighbour of vertex $i$, according to some arbitrary but fixed numbering of the outgoing edges of $i$. Following D{\" u}rr et al. \cite{durr2004quantum}, we assume that the degrees are provided for free as a part of the input, and we account only for queries to the arrays $f_i$. Moreover, we assume that the graph is not a multigraph (that is, each $f_i$ is injective).
\end{itemize}
We also assume that the input graph is undirected, and therefore the adjacency matrix is taken to be symmetric.

Classically, in the adjacency matrix model, each of the problems requires $\Theta(n^2)$ queries to the input adjacency matrix, since both the randomised and deterministic query complexities of any (non-trivial) subgraph-closed graph property are $\Theta(n^2)$, which can be shown via a reduction from the unstructured search problem \cite{childs_quantum_2012}. Known classical algorithms that achieve this bound are based on breadth first search, and hence require more than logarithmic space; however, by allowing more time the space requirement can be reduced. In particular, by using random walks, Aleliunas et al. \cite{aleliunas_random_1979} provide $O(n^3)$ time and $O(\log n)$ space algorithms for deciding bipartiteness and s-t connectivity. By taking into account a reduction given in this paper, this implies a similar algorithm for detecting arbitrary cycles. 

We describe a bounded-error quantum algorithm in the adjacency matrix model, which, given as input a graph $G$ with vertex set $V$ and edge set $E$, returns some vertex $v \in V$ that is a part of a cycle in $G$ if such a cycle exists, and otherwise returns false. The algorithm runs in time $\tilde{O}(n^{3/2})$ and requires $O(\log n)$ bits and qubits of storage, where the $\tilde{O}$ notation hides poly-logarithmic factors in $n$. Our algorithms are based on quantum walks and can hence be seen as quantum analogues of the approach of Aleliunas et al.~\cite{aleliunas_random_1979}. 

By a simple modification of the original algorithm, we obtain an algorithm that can decide whether or not a graph is bipartite (i.e. contains no odd-length cycles), and has the same time and space requirements. For both problems our algorithms are optimal up to poly-logarithmic factors, almost matching the $\Omega(n^{3/2})$ quantum query lower bounds for bipartiteness \cite{zhang2005power} and cycle detection \cite{childs_quantum_2012}. 

The main new technical ingredient of the algorithms is a reduction from the problem of cycle detection (or odd-length cycle detection, in the case of bipartiteness) to the problem of s-t connectivity in some ancillary graph. We then apply the span-program-based s-t connectivity algorithm of Belovs and Reichardt, introduced in \cite{belovs_span_2012}, to this ancillary graph without explicitly constructing it. Following this, we make use of a variant of Grover search to look for a vertex that makes up a part of a cycle in the input graph. We include a full proof of the efficiency of the s-t connectivity algorithm, which was omitted from \cite{belovs_span_2012}.

We then turn to the adjacency array model. D{\" u}rr et al. prove a tight $\Omega(n)$ quantum query lower bound for the s-t connectivity problem in the array model \cite{durr2004quantum}, which we extend to give a $\Omega(n)$ bound on the problems of deciding bipartiteness and cycle detection in the array model. By combining our reduction to s-t connectivity with the s-t connectivity algorithm in \cite{durr2004quantum}, it is possible to construct algorithms that decide bipartiteness and detect cycles in time $\tilde{O}(n)$. These algorithms are therefore optimal up to poly-logarithmic factors, but require $O(n \log n)$ space \cite{durr2004quantum}. In order to preserve space efficiency, we use a quantum walk based algorithm to decide s-t connectivity, but at the expense of time efficiency. Our quantum walk based algorithm takes time $\tilde{O}(n\sqrt{d_m})$, where $d_m$ is the maximum degree of any vertex in the graph.

\subsection{Previous Work}
D{\" u}rr et al.~\cite{durr2004quantum} previously gave quantum query lower bounds for some graph problems in both the adjacency matrix model and the adjacency array model, and in particular show that the quantum query complexity of testing connectivity between two vertices is $\Theta(n^{3/2})$ in the matrix model. In the following, unless explicitly stated, we will assume that any bounds given are applicable to the adjacency matrix model. Ambainis et al. \cite{ambainis2008quantum} show that planarity also has quantum query complexity $\Theta(n^{3/2})$, and Zhang \cite{zhang2005power} gave a lower bound of $\Omega(n^{3/2})$ for the problems of bipartiteness and perfect matching. More generally, Sun et al. \cite{sun2004graph} showed that all graph properties have quantum query complexity $\Omega(\sqrt{n})$, and gave a non-monotone property for which this lower bound is tight (up to polylogarithmic factors). 

An interesting graph property for which a tight lower bound has not been found is the $H$-subgraph containment problem. In the most general form of this problem, we are asked to determine whether or not the input graph contains the fixed graph $H$ as a subgraph. The best known lower bound for this property is only $\Omega(n)$. A special case is the property of containing a triangle, and the best known lower bound for this problem is also $\Omega(n)$. Le Gall \cite{le2014improved} has described a quantum query algorithm that detects triangles using $O(n^{5/4})$ queries. Under the promise that the input graph either contains a triangle as a subgraph, or does not contain it as a minor (the so-called \emph{subgraph/not-a-minor problem}), Belovs and Reichardt \cite{belovs_span_2012} provide an $O(n)$ quantum query algorithm based on span-programs, which can also be implemented time efficiently. Also under the promise of the subgraph/not-a-minor problem, Wang \cite{wang_span-program-based_2013} gives a span-program-based algorithm capable of detecting a given tree as a subgraph in $\tilde{O}(n)$ time.

Monotone graph properties are those that are subgraph-closed -- that is, every subgraph of a graph with the property also has that property. Likewise, a graph property is minor-closed if every graph minor of a graph with the property also has that property. Over a series of papers, Robertson and Seymour \cite{robertson_graph_2004} showed that all minor-closed graph properties can be described by a finite set of forbidden minors -- graphs that do not appear as a minor of any graph possessing the property. Some minor-closed properties can also be characterised by a finite set of forbidden \emph{subgraphs}.

The widely believed Aanderaa-Karp-Rosenberg conjecture \cite{rosenberg1973time} states that the deterministic and randomised query complexities of all monotone graph properties are $\Theta(n^2)$. Childs and Kothari \cite{childs_quantum_2012} show that all minor-closed properties that cannot be characterised by a finite set of forbidden \emph{subgraphs} have quantum query complexity $\Theta(n^{3/2})$. On the other hand, they show that all minor-closed properties and sparse graph properties that can be characterised by finitely many forbidden subgraphs can be determined in $o(n^{3/2})$ queries. Reichardt and Belovs \cite{belovs_span_2012} extended this result to show that any minor-closed property that can be characterised by exactly one forbidden subgraph, which must necessarily be a path or a (subdivided) claw, has query complexity $O(n)$. 
\\
\\
Some of these previous results can be applied to finding cycles. In particular, Childs and Kothari~\cite{childs_quantum_2012} describe a method that can be used to reject graphs with more than $n$ edges in time $\tilde{O}(\sqrt{n})$. Since an $n$-vertex graph with more than $n$ edges must necessarily contain a cycle, we can immediately dismiss these cases in $\tilde{O}(\sqrt{n})$ time. The graphs that are not rejected are then guaranteed to have fewer than $n$ edges, which can be reconstructed using $O(n^{3/2})$ queries to the input adjacency matrix. Now we have (the adjacency matrix of) a graph that is promised to have fewer than $n$ edges. Under this promise, even running a classical algorithm such as breadth first search can determine whether or not this graph contains a cycle in $O(n)$ time. Overall, the process takes time $\tilde{O}(n^{3/2})$. However, in order to reconstruct the edges of the graph, we require coherently addressable access to $O(n \log n)$ classical bits/qubits.

To decide whether or not a graph is bipartite, \={A}ri\c{n}\v{s} \cite{arins_span-program-based_2015} designed a span program that gives rise to an (optimal) $O(n^{3/2})$ quantum query algorithm. Our algorithm is inspired by this span program, and makes use of the s-t connectivity span program of Belovs and Reichardt as a sub-routine in a similar manner. We also provide a time (and space) efficient implementation. 

Piddock \cite{piddock} described a span-program-based quantum query algorithm for detecting cycles of constant fixed length, subject to the subgraph/not-a-minor promise, which requires $O(n^{3/2})$ queries to the input for odd length cycles, and $O(n)$ queries for even length cycles.   This is in contrast to the present work, which detects cycles of arbitrary length (with no promise on the input). Within the adjacency array model, D{\" u}rr et al.~\cite{durr2004quantum} suggest a quantum query algorithm for deciding bipartiteness in $O(n)$ queries, and using $O(n \log n)$ space. An example of a reduction to s-t connectivity given in terms of span programs is the work of Jeffery and Kimmel~\cite{jeffery2015nand}, in which the problem of evaluating {\sc nand}-trees is reduced to the problem of s-t connectivity on certain graphs.

The algorithms presented in this paper achieve optimal time complexity up to poly-logarithmic factors, but require only $O(\log n)$ classical and quantum bits of storage. Thus, we emphasise that our algorithms are also space efficient, with respect to the number of classical and quantum bits of storage that they require. We assume that we have access to quantum RAM (QRAM \cite{giovannetti2008quantum}), and that we use this for storage. We assume that we are given access to an oracle that lets us evaluate edges of $G$, and which isn't counted against the space bound. Therefore our measure of space efficiency differs somewhat to other, alternative definitions: for example, in investigating the computational power of space-bounded quantum Turing machines, Watrous \cite{watrous1999quantum, watrous1999space} measured the space requirements of the quantum (and classical) Turing machines in terms of the number of bits required to encode certain information regarding configurations of these machines. Instead, we consider the size of the QRAM required for our algorithms to run.

\subsection{Organisation}
We begin by introducing some useful results and background material in section \ref{sec:preliminaries}. In section \ref{sec:reduction}, we present a reduction of the problem of cycle detection in a graph $G$ to the problem of s-t connectivity in some ancillary graph that is constructed from $G$, which is the main ingredient for the algorithms that follow. Section \ref{sec:algo} presents a randomised algorithm for deciding whether a given vertex in a graph is a part of a cycle, and discusses the probability with which this algorithm fails. Section \ref{sec:grover} describes a more general algorithm that allows the detection of arbitrary cycles, and then section \ref{sec:bipartite} explains how to use a modified version of this algorithm to decide whether or not a graph is bipartite. Finally, section \ref{sec:adj_array} discusses how to obtain an efficient algorithm in the adjacency array model, by using a quantum walk in place of the span-program-based s-t connectivity algorithm used in the previous sections. 

Appendices \ref{sec:span_programs_main} through \ref{sec:impl_sp} describe the span-program-based s-t connectivity algorithm of Belovs and Reichardt \cite{belovs_span_2012}, which is crucial for our results and introduced in section \ref{sec:preliminaries}. We include a proof of its correctness and time and space complexity, the details of which were omitted from \cite{belovs_span_2012}. Appendix \ref{sec:span_programs_main} briefly introduces span programs, and then Appendix \ref{sec:stspan} presents a span program for the problem of s-t connectivity. Appendix \ref{sec:impl_sp} describes a general method for implementing span programs time efficiently, due to Belovs and Reichardt, and then applies this method to the span program for s-t connectivity (following the approach in \cite{belovs_span_2012} very closely). Finally, Appendix \ref{app:diff_ops} details the implementation of the operations required for the quantum walk algorithm for s-t connectivity in the adjacency array model. 

%%%%%%%%%%%%%%%%%%%%%%%%%%%%%%%%%%%%%%%%%%%%%%%
%%%%%%%%%%%%%%%%%%%%%%%%%%%%%%%%%%%%%%%%%%%%%%%
%%%%%%%%%%%%%%%%%%%%%%%%%%%%%%%%%%%%%%%%%%%%%%%
%								PRELIMS
%%%%%%%%%%%%%%%%%%%%%%%%%%%%%%%%%%%%%%%%%%%%%%%
%%%%%%%%%%%%%%%%%%%%%%%%%%%%%%%%%%%%%%%%%%%%%%%
%%%%%%%%%%%%%%%%%%%%%%%%%%%%%%%%%%%%%%%%%%%%%%%
\section{Preliminaries}\label{sec:preliminaries}

We will make use of the following result of Belovs and Reichardt \cite{belovs_span_2012}:

\begin{restatable}{theorem}{therbelovs}[Combination of Theorems 3 and 9 from \cite{belovs_span_2012}]
\label{ther:belovs}
Consider the st-connectivity problem on a graph G given by its adjacency matrix. Assume there is a promise that if s and t are connected by a path, then they are connected by a path of length at most $d$. Then there exists a bounded-error quantum algorithm that determines whether $s$ and $t$ are connected in $\tilde{O}(n\sqrt{d})$ time and uses $O(\log n)$ bits and qubits of storage, and which fails with probability at most $1/10$. 
\end{restatable}
The proof of this theorem can be found in Appendix \ref{app:stconn}.\\
\\
We will also require some facts about \emph{k-wise independent hash functions}:
\begin{definition}[\cite{mitzenmacher2005probability}]
Let $U$ be a universe with $|U| \geq n$ and let $V = \{0,1,...,n-1\}$. A family of hash functions $\mathcal{H}$ from $U$ to $V$ is said to be \emph{strongly k-universal} if, for any elements $x_1, x_2, ..., x_k \in U$, any values $y_1, y_2, ..., y_k \in V$, and a hash function $h$ chosen uniformly at random from $\mathcal{H}$, we have
\[
\Pr[(h(x_1) = y_1) \cap (h(x_2) = y_2) \cap \cdots \cap (h(x_k) = y_k)] = \frac{1}{n^k}.
\]
\end{definition}
We will be interested in the case where $k=2, n=2$. In this case, the values $h(x_1), h(x_2)$ are pairwise independent, since the probability that they take on any pair of values is $\frac{1}{n^2} = \frac{1}{4}$. The simplest construction, which suffices for our purposes, is to use functions $h: \{0,1\}^m \rightarrow \{0,1\}$ of the form $h(x) = \braket{a,x} + b \mod 2$, where $\braket{a,x} = \sum_{i=1}^m a_ix_i \mod 2$ \cite{mitzenmacher2005probability}. Each function is parameterised by two values $a \in \{0,1\}^m, b \in \{0,1\}$. To achieve pairwise independence of the values $h(x_1), h(x_2)$ for $x_1, x_2 \in \{0,1\}^m$, we therefore require $m+1$ truly random bits to specify $a$ and $b$. Doing so gives us $N = 2^m$ pairwise independent `random' bits. 

To use a hash function to assign a value in $\{0,1\}$ to each of $N$ elements, we require $O(\log N)$ bits to specify the hash function $h$, from which $h(x)$ can be calculated in $O(\log^2 N)$ time \cite{mitzenmacher2005probability}.

Throughout the paper, we will use $[n] := \{1,...,n\}$ to denote the integers from $1$ to $n$.

%%%%%%%%%%%%%%%%%%%%%%%%%%%%%%%%%%%%%%%%%%%%%%%
%%%%%%%%%%%%%%%%%%%%%%%%%%%%%%%%%%%%%%%%%%%%%%%
%%%%%%%%%%%%%%%%%%%%%%%%%%%%%%%%%%%%%%%%%%%%%%%
%								REDUCTION
%%%%%%%%%%%%%%%%%%%%%%%%%%%%%%%%%%%%%%%%%%%%%%%
%%%%%%%%%%%%%%%%%%%%%%%%%%%%%%%%%%%%%%%%%%%%%%%
%%%%%%%%%%%%%%%%%%%%%%%%%%%%%%%%%%%%%%%%%%%%%%%
\section{Reduction of Cycle Detection to s-t Connectivity}\label{sec:reduction}
Let $G = (V,E)$ be a connected, undirected graph on $n$ vertices. Fix some arbitrary orientation of the edges $(u,v) \in E$ by directing edges from $u \rightarrow v$ if $v>u$, $v \rightarrow u$ otherwise. $G$ is now a directed graph. 

Now consider an ancillary graph $H = (V',E')$, where $V' = \{s,t\} \cup \{v_b : v \in V, b \in \{0,1,2\}\}$ and $E' = \{(u_b,v_{b+1} \mod 3) : (u,v) \in E, b \in \{0,1,2\}\} \cup \{(s,k_0), (t,k_1)\}$ for some $k \in V$. Intuitively, we split each vertex $v \in V$ into three vertices $v_0, v_1$, and $v_2$. Then, for each (directed) edge $(u,v) \in E$, we create three edges $(u_0,v_1), (u_1,v_2)$, and $(u_2,v_0)$ in $H$. Finally, we add an edge between $s$ and $k_0$ and between $t$ and $k_1$, for some arbitrarily chosen vertex $k$. 

%The structure of $H$ is shown in Figure \ref{fig:H_general}, in which we have chosen $k=1$. 
%
%\begin{figure}[htbp]
%\begin{center}
%\includegraphics[scale = 0.5]{diagrams/H_general.eps}
%\caption{The structure of the graph $H$}
%\label{fig:H_general}
%\end{center}
%\end{figure}

To analyse the reduction to s-t connectivity, we introduce the notion of `clockwise' and `anticlockwise' edges. Given an undirected cycle, fix an arbitrary vertex $v$ in the cycle that has at least 1 outgoing edge that makes up a part of the cycle (it is easy to verify that such a vertex must exist). Starting with one of the outgoing edges, we traverse the cycle from $v$ back to $v$. Any edge that is oriented in the direction of traversal is defined as clockwise, and any edge oriented against the direction of traversal is defined as anticlockwise. Figure \ref{fig:H_example} provides an example to illustrate the notion of clockwise and anticlockwise edges, and gives two examples of the form of the graph H constructed from a cycle on 4 vertices, showing how the reduction fails when the number of clockwise and anticlockwise edges is equal modulo 3. 

%\begin{figure}[htbp]
%\begin{center}
%\includegraphics[scale = 0.5]{diagrams/clockwise.eps}
%\caption{An example of a cycle with clockwise (solid) and anticlockwise (dashed) edges}
%\label{fig:clockwise}
%\end{center}
%\end{figure}

\begin{figure}[htbp]
\begin{center}
\includegraphics[scale = 0.575]{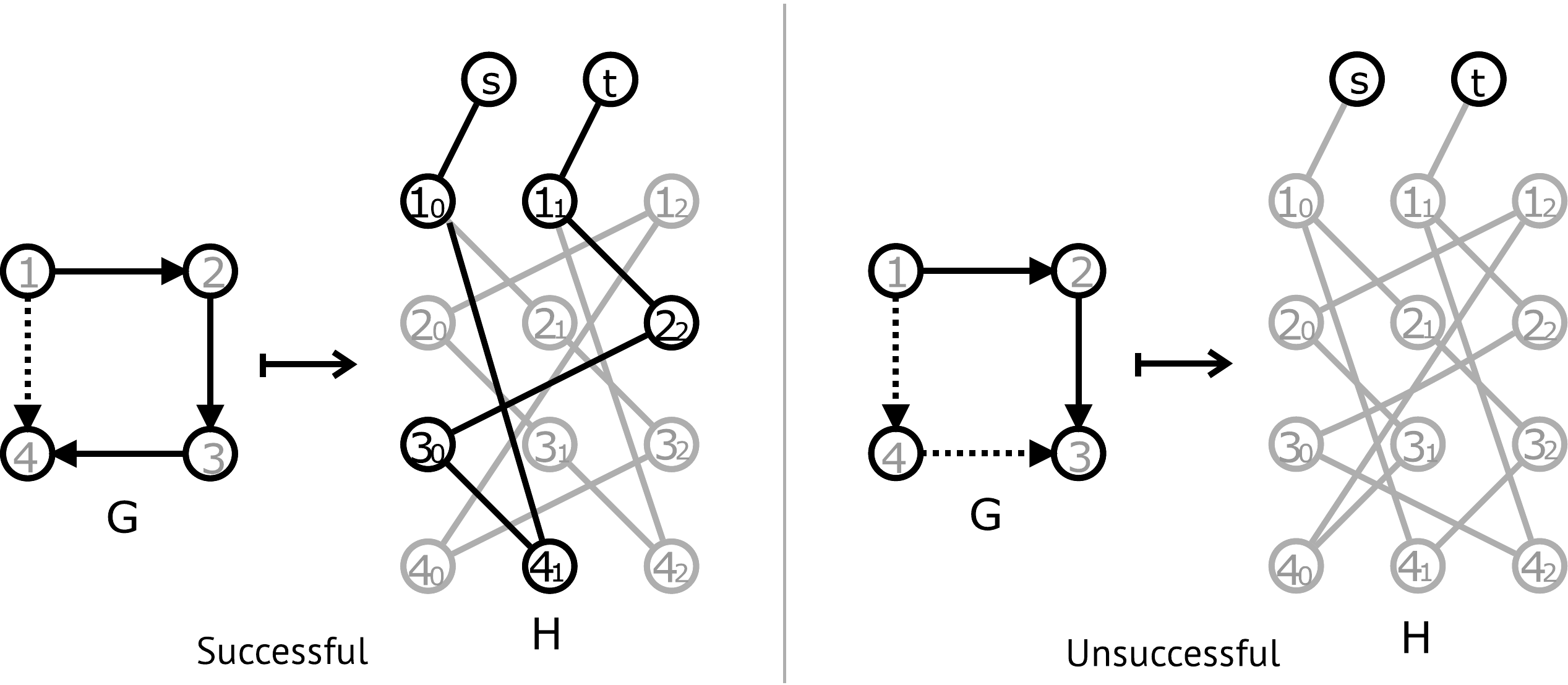}
\caption{A successful (left) and unsuccessful (right) reduction to s-t connectivity for $G=C_4$, a cycle on 4 vertices, with clockwise and anticlockwise edges represented by solid and dashed arrows, respectively}
\label{fig:H_example}
\end{center}
\end{figure}

We now prove the following lemma:

\begin{lemma}\label{lem:reduction}
Let $G=(V,E)$ be a connected undirected graph, and let $H = (V',E')$ be defined as above. Then there is a path from $s$ to $t$ in $H$ if and only if there is a cycle present in $G$, such that the difference $D := p - q$ between the number of clockwise edges $p$ and anticlockwise edges $q$ satisfies $D \not\equiv 0 \mod 3$. Furthermore, if $k$ is chosen to be a vertex on the cycle, then the length of the path between $s$ and $t$ is at most $2c+2$, where $c$ is the length of the cycle. 
\end{lemma}

\begin{proof}
First we show that if there is a cycle in $G$ such that $D \not\equiv 0 \mod 3$, then there is also a path from $s$ to $t$ in $H$. We will assume that the vertices of $G$ are labelled arbitrarily by the integers $1..n$, and (without loss of generality) that $k=1$. To find a path from $s$ to $t$, it suffices to find a path from $1_0$ to $1_1$. We assume that the edges of the cycle are oriented in such a way that $D \not\equiv 0 \mod 3$. It is useful to recall that all edges in $H$ are of the form $(u_b, v_{b+1 \mod 3})$, for $b \in \{0,1,2\}$, and such an edge only exists if the edge $(u,v)$ is present in $G$. Suppose that the cycle is of length $c$, and is composed of the vertices $2,3,4,...,c,2$, where each vertex label is arbitrary. First we show that this implies that there is a path from $2_0$ to $2_1$ in $H$. In fact, we prove something stronger: that there must exist a cycle of length $3c$ in $H$ that contains all vertices $i_b$ for $i \in \{2,...,c\}, b\in \{0,1,2\}$. 

To see this, suppose that, starting at the vertex $2_0$ in $H$, we follow the edges of $H$ that correspond to the edges of the cycle in $G$. Depending on the orientation of the first edge, we first move to either vertex $3_1$ (if the edge is directed $2\rightarrow 3$) or $3_2$ (if the edge is directed $3\rightarrow 2$). In general, at each step, we move from a vertex $u_d$ to a vertex $u_{d\pm1 \mod 3}$, where the clockwise edges add 1 to the value of $d$, and the anticlockwise edges add $-1$. We will refer to the value of $d$ as the `parity' of the vertex. After taking $c$ steps, we will have traversed $p$ clockwise edges and $q$ anticlockwise edges, and so we will arrive at vertex $2_{p-q \mod 3}$. If $p - q = D \not\equiv 0 \mod 3$, then we must be at either $2_1$ or $2_2$, depending on the value of $D$. By traversing the cycle again, we arrive at vertex $2_{2D \mod 3}$. Traversing the cycle one final time, we arrive at $2_{3D \mod 3} = 2_0$. Since $2D \not\equiv D \mod 3$, unless $D \equiv 0 \mod 3$, vertices $2_{D \mod 3}$ and $2_{2D \mod 3}$ are distinct. Therefore, we have a path from $2_0$ to $2_b$ for some $b \in \{1,2\}$, from $2_b$ to $2_{b'}$ for $b' \in \{1,2\}\setminus\{b\}$, and from $2_{b'}$ to $2_0$. Each path must necessarily be disjoint, since each traversal around the edges of the cycle in $G$ adds the same sequence of $+1$s and $-1$s to the parity, and therefore starting with a different initial parity ensures a unique path through the vertices of $H$. Since we have 3 disjoint paths of length $c$, combining them gives us a cycle of length $3c$ that includes all vertices in $H$ corresponding to vertices in $G$ that make up the cycle. 

A straightforward consequence of this is that there exists a path from $2_b$ to $2_{b'}$ in $H$ if there exists a cycle in $G$ containing the vertex $2$, for $b \neq b' \in \{0,1,2\}$. Since $G$ is connected, there must be a path from $1$ to $2$ in $G$. By following the edges of this path, we can find a corresponding path in $H$ from $1_0$ to $2_b$ and from $1_1$ to $2_{b+1 \mod 3}$, for some $b \in \{0,1,2\}$. Since there exists a path from $2_b$ to $2_{b+1\mod 3}$ in $H$, there must also exist a path from $1_0$ to $1_1$. 

The length of this path will depend upon two things: the length of the shortest path from $1$ to $2$ in $G$, and the length of the cycle in $G$. The former is determined by the length the path from $1$ to $2$. In particular, following the edges of $H$ that correspond to this path in $G$ will lead paths of length  from $1_0$ to $2_b$, and from $1_1$ to $2_{b'}$, for $b\neq b' \in \{0,1,2\}$. The length of the path from $2_b$ to $2_{b'}$ in $H$ is then at most $2c$. To see this, consider traversing the edges of $H$ corresponding to the cycle in $G$, starting at $2_b$. As argued above, after $c$ steps we will arrive at a vertex $2_{\tilde{b}}$, for $\tilde{b}\neq b \in \{0,1,2\}$. If $\tilde{b} = b'$, then the length of the path is $c$. On the other hand, if $\tilde{b} \neq b'$ then we can take $c$ more steps, at which point we will arrive at $2_{b'}$ after a total of $2c$ steps. In order to prove the final part of the lemma, we note that when the vertex $k$ (which we have assumed without loss of generality is the vertex 2 in this case) is contained in the cycle, then the path from $s$ to $t$ will be determined only by the length of the path from $2_0$ to $2_1$. By the arguments given here, this is at most of length $2c$. Adding in the two edges incident to vertices $s$ and $t$, we obtain the upper bound of $2c+2$. 
\newline
\\
We will now show that if $G$ does not contain a cycle, or if it contains a cycle such that $D = p-q  \equiv 0 \mod 3$, then $s$ and $t$ are not connected in $H$. 

Assume that $G$ does not contain a cycle at all. Suppose that there is a path $P$ from $v_0$ to $v_1$ in $H$, of the form $v_0, v'_{b}, v''_{b'}, ..., v_1$. Since, for every edge $(u_b, v_{b'}) \in E'$ there must be a corresponding edge $(u,v) \in E$, we can construct a path $Q=v, v', v'', ..., v$ in $G$ from $P$. However, this gives a cycle in $G$, and hence, a contradiction.

Suppose instead that there is a cycle of length $c$ in $G$ such that $D \equiv 0 \mod 3$. We have shown that this implies the existence of a path of length $c$ from $v_b$ back to $v_b$ for any vertex $v$ in the cycle and for all $b \in \{0,1,2\}$. Since each such cycle must be disjoint (by the same argument as before), there cannot be a path from any $v_b$ to $v_{b'}$ for $b \neq b'$, else the cycles would necessarily share vertices. 
\end{proof}

Lemma \ref{lem:reduction} shows us that the problem of cycle detection on a graph $G$ reduces to the problem of s-t connectivity on some ancillary graph $H$. Therefore, if we can test for s-t connectivity efficiently on the graph $H$, then we can also test efficiently for cycles in $G$. However, this reduction fails when the input graph is such that the number of clockwise edges equals the number of anticlockwise edges modulo 3. The next section will discuss a randomised algorithm that deals with this case.

%%%%%%%%%%%%%%%%%%%%%%%%%%%%%%%%%%%%%%%%%%%%%%%
%%%%%%%%%%%%%%%%%%%%%%%%%%%%%%%%%%%%%%%%%%%%%%%
%%%%%%%%%%%%%%%%%%%%%%%%%%%%%%%%%%%%%%%%%%%%%%%
%								Algorithm
%%%%%%%%%%%%%%%%%%%%%%%%%%%%%%%%%%%%%%%%%%%%%%%
%%%%%%%%%%%%%%%%%%%%%%%%%%%%%%%%%%%%%%%%%%%%%%%
%%%%%%%%%%%%%%%%%%%%%%%%%%%%%%%%%%%%%%%%%%%%%%%
\section{Algorithm for Cycle Detection}\label{sec:algo}
In this section, we describe an algorithm that makes use of both the reduction of cycle detection to s-t connectivity, and the s-t connectivity algorithm of Belovs and Reichardt. In particular, we prove:

\begin{theorem}\label{theo:main_algo}
There exists a quantum algorithm which, given as input a graph $G=(V,E)$, a vertex $k\in V$, and an integer $d$, outputs true with probability $\geq 9/20$ if $G$ contains a cycle of length $l \leq d$ that includes $k$, and returns false with probability $\geq 9/10$ if it does not contain any cycle. The algorithm takes $\tilde{O}(n\sqrt{d})$ time and requires $O(\log n)$ space. 
\end{theorem}

\begin{proof}
By Lemma \ref{lem:reduction}, the problem of detecting if a vertex $k$ is included in a cycle reduces to the problem of s-t connectivity on an ancillary graph $H$, which is constructed from $G$. By Theorem \ref{ther:belovs}, we can test for s-t connectivity in an $n$ vertex graph in $\tilde{O}(n\sqrt{l})$ time and $O(\log n)$ space, where $l$ is an upper bound on the length of the path connecting $s$ and $t$. Lemma \ref{lem:reduction} also states that, if $k$ is contained within a cycle of length $d$, then there is a path from $s$ to $t$ in $H$ of length at most $2d+2$. It is worth noting that Lemma \ref{lem:reduction} actually gives a stronger result -- that there is a path from $s$ to $t$ in $H$ when there exists some path from $k$ to a cycle in $G$, provided that the cycle satisfies some constraints on the orientations of its edges. However, as we will see, for some inputs, the algorithm could fail to detect such cases with certainty. Thus, we restrict ourselves to the worst case -- that in which the algorithm can only detect the presence of a cycle that includes the vertex $k$. 

If we can show that there is some efficient map from the edges of $H$ to the edges of $G$, then the s-t connectivity algorithm can be run on the graph $H$ whilst only querying the input oracle for $G$, and we can detect cycles in $G$ in time $\tilde{O}(n\sqrt{d})$.

Additionally, we must show that the algorithm fails only with some constant probability, even when given a `bad' input (one with an equal number of clockwise and anticlockwise edges (modulo 3)). 

We begin by describing a randomised approach which causes the reduction to s-t connectivity to fail with only constant probability when given a bad input. Suppose we were to run the s-t connectivity algorithm on the graph $H$ associated with an input graph $G$, which contains a cycle. If the adjacency matrix of the graph were such that the number of clockwise edges and the number of anticlockwise edges on the cycle were congruent modulo 3, then the algorithm, as it stands, would fail to detect the cycle with certainty, as a result of Lemma \ref{lem:reduction}.

We could prevent the algorithm from failing by flipping the direction of a single edge on the cycle. Recall that we are given some vertex $k \in V$, and add edges $(s,k_0)$ and $(t,k_1)$ in $H$ before testing for a path between $s$ and $t$. Our solution is to flip the direction of some random subset of the edges adjacent to vertex $k$, and show that this flips exactly one edge on the cycle with high probability ($\geq 1/2$).

To choose a random subset of edges to flip, we colour every vertex in $G$ with a colour chosen from $\{0,1\}$. Then, for every edge adjacent to $k$ in $G$, if the vertex at the end of the edge is coloured $1$, we flip the direction of the edge, and otherwise do nothing. The colouring is achieved using a family of pairwise independent hash functions $\mathcal{H}$ from $[n]$ to $\{0,1\}$. The pairwise independence gives the constraint that, for $x,y \in [n]$ and $a,b \in \{0,1\}$, and a hash function $h$ chosen uniformly at random from $\mathcal{H}$,
\[
\Pr[h(x)=a \cap h(y)=b] = \frac{1}{4}.
\]
Let the two vertices adjacent to $k$ in the cycle be $a$ and $b$, and choose a pairwise independent hash function $h$ uniformly at random from $\mathcal{H}$. Then, by pairwise independence, we have
\[
\Pr[h(a)=h(b)=0]=\Pr[h(a)=h(b)=1]= \frac{1}{4} 
\]
and
\[
\Pr[h(a) \neq h(b)]= \frac{1}{2}.
\]
So the above method will fail to flip either of the edges $(a,k), (b,k)$ with probability $\frac{1}{4}$. Otherwise, with probability $\frac{1}{2}$ exactly one of the two edges will be flipped, and with probability $\frac{1}{4}$ both edges will be flipped. Thus, with probability at least $\frac{1}{2}$, the number of clockwise edges will no longer equal the number of anticlockwise edges modulo 3. Therefore, by colouring the vertices of the graph using a pairwise independent hash function, the algorithm will fail with probability at most $\frac{1}{2}$ when given a `bad' input. On the other hand, if the algorithm is given some `good' input, then this process may cause the algorithm to fail; however, this will happen with probability at most $1/4$, by the same argument as before.
%It should be noted that this process does change the behaviour of the algorithm somewhat, although this change makes little difference to the overall algorithm. Lemma \ref{lem:reduction} states that, so long as the input graph is connected, a cycle can be detected regardless of our choice of vertex $k$ (i.e. the vertex that ends up being connected to $s$ and $t$ in $H$). That is, the algorithm accepts if there is a cycle present in the same connected subcomponent of $G$ that contains $k$, and rejects otherwise, failing to detect any cycle with an equal number of clockwise and anticlockwise edges (modulo 3). By using the colouring method described in this section, the algorithm will only fail to detect such a cycle with probability at most $1/2$. However, we lose some freedom over the choice of the vertex $k$. In particular, the algorithm takes as input a graph $G=(V,E)$ and a vertex $k \in V$, and accepts with probability at least $1/2$ if there is a cycle in $G$ which \emph{contains vertex $k$}, and rejects otherwise.
\\
\\
Now we consider the map from the edges of $H$ to the edges of $G$. We can query the adjacency matrix of $H$ implicitly by querying the entries of $G$'s adjacency matrix. Given two vertices $u_b$ and $v_{b'}$ in $H$, we test for the presence of the edge $(u_b, v_{b'})$ as follows. First we determine the direction of the edge $(u,v)$ in $G$, if it were to exist. If $v > u$, then the edge is directed from $u \rightarrow v$, otherwise it is directed from $v \rightarrow u$. If $u=k$, then we look up the colour of vertex $v$ as determined by our hash function, and vice versa if $v=k$. If the colour is 0, we do nothing; if it is 1, we flip the edge. 

Next we test whether the edge is allowed to exist. If the edge is directed from $u \rightarrow v$, then it is allowed only if $b' \equiv b+1 \mod 3$. Similarly, if the edge is directed from $v \rightarrow u$, then it is allowed only if $b' \equiv b-1 \mod 3$. Finally, if the edge is allowed to exist, then we test for its presence in $G$ by querying the $uv$ entry of $G$'s adjacency matrix. If the result of the query is 1, then the edge exists and we return 1. In all other cases (the query returns 0, or the edge is not allowed), we return 0.

We can use this map to run the s-t connectivity algorithm on $H$ without explicitly constructing it. That is, rather than allowing the algorithm to query the input oracle for $G$, we allow it to query the circuit that implements the process described above, which will query the input oracle for $G$ as appropriate. All steps of the map run in time $\text{polylog} (n)$ and require $\log n$ space.

%%%%%%%%%%%%%%%%%%%%%%%%%%%%%%%%%%%%%%%%%%%%%%%
%%%%%%%%%%%%%%%%%%%%%%%%%%%%%%%%%%%%%%%%%%%%%%%
%%%%%%%%%%%%%%%%%%%%%%%%%%%%%%%%%%%%%%%%%%%%%%%
%							FAILURE PROBABILITIES
%%%%%%%%%%%%%%%%%%%%%%%%%%%%%%%%%%%%%%%%%%%%%%%
%%%%%%%%%%%%%%%%%%%%%%%%%%%%%%%%%%%%%%%%%%%%%%%
%%%%%%%%%%%%%%%%%%%%%%%%%%%%%%%%%%%%%%%%%%%%%%%
\paragraph{Probability of Failure --}
We have already shown that if the input graph $G$ contains a cycle that includes vertex $k$, then there will be a path from $s$ to $t$ in $H$ with probability at least $1/2$. In this section, we consider the probability of detecting this path using the s-t connectivity algorithm of Belovs and Reichardt. Recall that their algorithm runs in time $\tilde{O}(n\sqrt{d})$, where $d$ is an upper bound on the length of the path connecting $s$ and $t$, and requires $O(\log n)$ space. In our case, the length of the path is proportional to the length of the cycle that we are trying to detect. Since we do not know this length in advance, we do not know an upper bound on the length of the path between $s$ and $t$ in $H$. Therefore, we will have to `guess' an upper bound, and modify this guess as the algorithm progresses. 

By Theorem \ref{ther:belovs}, if $d \geq l$, then the s-t connectivity algorithm detects the presence of a cycle with probability at least $9/10$ if one exists, and otherwise says that no cycle exists with probability at least $9/10$. If $G$ contains a cycle of length $l \leq d$ that includes $k$, then with probability $p \geq 1/2$ there will be a path from $s$ to $t$ in $H$. By running the s-t connectivity algorithm on $H$, we will detect this path with probability $9/20$.

If $G$ does not contain a cycle, then there will be no path from $s$ to $t$ in $H$, and the s-t connectivity algorithm will return false with probability $\geq 9/10$. 

In the case that $G$ contains a cycle of $l > d$, then the s-t connectivity algorithm may still detect the presence of the path from $s$ to $t$ in $H$. However, since this probability could be very small, we shall assume that it never detects such a cycle. 

The colouring step requires the use of a pairwise-independent hash function, which requires $O(\log n)$ space and $O(\log^2 n)$ time. The s-t connectivity algorithm requires $O(\log n)$ space and $\tilde{O}(n\sqrt{d})$ time. Therefore, the algorithm of Theorem \ref{theo:main_algo} requires $\tilde{O}(n\sqrt{d})$ time and $O(\log n)$ space.

\end{proof}

%%%%%%%%%%%%%%%%%%%%%%%%%%%%%%%%%%%%%%%%%%%%%%%
%%%%%%%%%%%%%%%%%%%%%%%%%%%%%%%%%%%%%%%%%%%%%%%
%%%%%%%%%%%%%%%%%%%%%%%%%%%%%%%%%%%%%%%%%%%%%%%
%								GROVER SEARCH
%%%%%%%%%%%%%%%%%%%%%%%%%%%%%%%%%%%%%%%%%%%%%%%
%%%%%%%%%%%%%%%%%%%%%%%%%%%%%%%%%%%%%%%%%%%%%%%
%%%%%%%%%%%%%%%%%%%%%%%%%%%%%%%%%%%%%%%%%%%%%%%
\section{Detecting Arbitrary Cycles}\label{sec:grover}
In the previous sections we described an algorithm that, given a graph $G=(V,E)$, a vertex $k \in V$, and an estimate $d$ of the length of a cycle in $G$, outputs 1 with probability $\geq 9/20$ if there is a cycle of length $l \leq d$ in $G$ that contains $k$, and outputs $0$ with probability $\geq 9/10$ if $G$ does not contain a cycle. By repeating the algorithm $O(\log 1/\epsilon)$ times and using majority voting, we obtain an algorithm $\mathcal{A}$ that fails (i.e. returns false positives or false negatives) with probability at most $\epsilon$. In particular, we could reduce the probability of failure of $\mathcal{A}$ to $\frac{1}{\text{poly}(n)}$ with only an $O(\log n)$ overhead. In this case, since the overall algorithm calls $\mathcal{A}$ $\text{poly}(n)$ times, we can reduce the probability of failure of the overall algorithm to an arbitrary constant. This holds even for quantum algorithms calling $\mathcal{A}$ in superposition \cite{bernstein1997quantum, hoyer2003quantum}.

We can use this algorithm as a sub-routine for a more general algorithm that is capable of detecting the presence of arbitrary cycles in $G$. In particular, we make use of a variant of Grover search over a set of $N$ elements that allows us to search for a good solution without knowing how many good solutions there are. The approach was introduced in \cite{boyer1996tight}, and proceeds as follows:

\paragraph{QSearch:}
\begin{enumerate}
\item Initialise $m=1$ and set $\lambda$ so that $1 < \lambda < 4/3$.
\item Choose $j$ uniformly at random from the nonnegative integers smaller than $m$.
\item Apply $j$ iterations of Grover's algorithm, starting from initial state $\ket{\Psi_0} = \sum_v \frac{1}{\sqrt{N}}\ket{v}$.
\item Observe the register: let $i$ be the outcome.
\item If $i$ is indeed a solution, then the problem is solved: exit.
\item Otherwise, set $m$ to $\min(\lambda m, \sqrt{N})$ and go back to step 2.
\end{enumerate}
~\\
Then \cite{boyer1996tight} proves the following result:

\begin{theorem}[Theorem 3 of \cite{boyer1996tight}]\label{theo:grover}
Given oracle access to some Boolean function $f : [N] \rightarrow \{0,1\}$, such that the set of `solutions' $M = \{x \in [N] : f(x) = 1\}$ has unknown size $t = |M|$, the algorithm {\bf QSearch} finds a solution if there is one using an expected number of $O(\sqrt{N/t})$ Grover iterations. In the case that there is no solution, then {\bf QSearch} runs forever.
\end{theorem}

We will use a variant of this algorithm to search over the set of vertices in the input graph, using the algorithm $\mathcal{A}$ as an oracle -- that is, for each vertex $v \in V$, we will call $\mathcal{A}$ with some guess $d$ and with the vertex $k$ set to $v$. We will call {\bf QSearch} multiple times, each time with a different guess at the cycle length, and will ask it to stop after some number of iterations that depends on the current guess. More precisely, we perform the following:

\begin{enumerate}
\item For $i = 1$ to $\lceil \log_2 n \rceil$:
\begin{enumerate}
\item Run {\bf QSearch} over the vertices of the graph with $d = 2^i$, and stop when we have performed more than $C'' \sqrt{\frac{n}{2^i}}$ Grover iterations, for some constant $C''$.
\item If {\bf QSearch} returns a solution, then output the solution and exit, otherwise continue.
\end{enumerate}
\item Output `no cycle exists'.
\end{enumerate}

This detects the presence of a cycle with high probability if one exists, since, as soon as $i$ becomes large enough that $d = 2^i \geq l$, the cycle detection algorithm detects the presence of cycles with (arbitrarily) high probability. At this point, {\bf QSearch} finds a good solution (i.e. a vertex that causes the cycle detection algorithm to accept) with high probability.

More precisely, when $d \geq l$, $l$ out of $n$ vertices will provide good solutions (i.e. will have caused $\mathcal{A}$ to accept). At this point, by Theorem \ref{theo:grover}, the expected number of Grover iterations required to find a solution using {\bf QSearch} is $\leq D\sqrt{n/l}$ for some (known) constant $D$. When $i$ first becomes large enough that $d = 2^i \geq l$, the guess $d$ will be at most twice the length of the cycle $l$, i.e. $d \leq 2l$. Since the expected number of iterations to find a solution using {\bf QSearch} at this point is $\leq D\sqrt{n/l}$, if we stop after $C'' \sqrt{\frac{n}{2^i}}$ iterations, the probability that {\bf QSearch} has not been able to find a solution yet is bounded above by
\[
\Pr\left[\text{no solution after } C'' \sqrt{\frac{n}{2l}}\text{ iterations}\right] \leq \frac{D\sqrt{2}}{C''}
\]
by Markov's inequality, and we can choose $C''$ to make this probability arbitrarily small. Each round of the algorithm after this point has a smaller probability of detecting the cycle, since we perform fewer Grover iterations as $i$ increases. However, by choosing sufficiently large $C''$, we can ensure that the probability of failing to detect a cycle during the first round of {\bf QSearch} in which $d \geq l$ is smaller than, say, $1/3$. 

If $G$ does not contain a cycle, then $\mathcal{A}$ will return false with high probability on all vertices, independently of the value of $d$. Therefore, {\bf QSearch} will fail to `find' a solution with high probability every time it is run, and the above algorithm will output `no cycle exists' with high probability. 

To analyse the time complexity of the algorithm, we will consider what happens when there is no cycle present. If a cycle is present, then by the above argument it will be found with high probability and the algorithm will exit early, requiring less time. In the $i^{th}$ round of the above algorithm, we run {\bf QSearch}, stopping after $O(\sqrt{\frac{n}{2^i}})$ iterations of Grover search have been performed. Each iteration of Grover search requires a single call to both $\mathcal{A}$ and $\mathcal{A}^{-1}$, each of which take time $\tilde{O}(n\sqrt{d}) = \tilde{O}(n\sqrt{2^i})$. Therefore, the time taken to run the $i^{th}$ round is $O(\sqrt{\frac{n}{2^i}})\cdot\tilde{O}(n\sqrt{2^i}) = \tilde{O}(n^{3/2})$.

We run at most $\lceil \log_2 n\rceil$ rounds of the above process, requiring at most $\tilde{O}(n^{3/2})$ time in total. 
\\
\\
In summary, by making use of a variant of Grover search and repeatedly guessing at increasing cycle lengths, we are able to find a vertex $k \in V$ that is part of a cycle in $G$ with probability $\geq 2/3$ if such a vertex exists, and return false with probability $\geq 2/3$ if $G$ contains no cycle. This requires $\tilde{O}(n^{3/2})$ time and $O(\log n)$ bits and qubits of storage.

%%%%%%%%%%%%%%%%%%%%%%%%%%%%%%%%%%%%%%%%%%%%%%%
%%%%%%%%%%%%%%%%%%%%%%%%%%%%%%%%%%%%%%%%%%%%%%%
%%%%%%%%%%%%%%%%%%%%%%%%%%%%%%%%%%%%%%%%%%%%%%%
%								BIPARTITENESS
%%%%%%%%%%%%%%%%%%%%%%%%%%%%%%%%%%%%%%%%%%%%%%%
%%%%%%%%%%%%%%%%%%%%%%%%%%%%%%%%%%%%%%%%%%%%%%%
%%%%%%%%%%%%%%%%%%%%%%%%%%%%%%%%%%%%%%%%%%%%%%%
\section{Deciding Bipartiteness}\label{sec:bipartite}  
We can view the algorithm for cycle detection as a special case of a more general algorithm. Currently, we arbitrarily orient the edges of an initially undirected graph, and accept if this forms a cycle in which the number of clockwise and anticlockwise edges are unequal modulo 3. We might ask what happens when we look for cycles with an unequal number of clockwise and anticlockwise edges modulo some other constant $s$. This would change the reduction to s-t connectivity by modifying the structure of the graph $H$ that is constructed from $G$. In particular, each vertex in $G$ would be split into $s$ sub-vertices in $H$, and for each edge $(u,v)$ in $G$ we would have $s$ corresponding edges $(u_0, v_1), (u_1, v_2), ..., (u_{s-2},v_{s-1}), (u_{s-1}, v_0)$. The cases $s > 3$ behave similarly to the case $s=3$, and are uninteresting. However, the case $s=2$ proves useful. In this case, the original algorithm (i.e. without random colouring) fails when the number of clockwise edges and the number of anticlockwise edges in every cycle $C$ differs by some multiple of 2. This will be the case if $C$ is of even length, and is independent of the orientation of the individual edges. Conversely, no odd-length cycle will cause the algorithm to `fail'. This means that, given some graph $G=(V,E)$ and a vertex $k \in V$, the algorithm will accept (with certainty) if $k$ is a part of an odd-length cycle, and reject otherwise. A graph is bipartite if and only if it contains no odd-length cycles. Thus, setting $s=2$ (and omitting the colouring step) allows us to decide whether or not a graph is bipartite. Since we have only removed a step of the algorithm, it still runs in time $\tilde{O}(n^{3/2})$ and requires $O(\log n)$ space.

%%%%%%%%%%%%%%%%%%%%%%%%%%%%%%%%%%%%%%%%%%%%%%%
%%%%%%%%%%%%%%%%%%%%%%%%%%%%%%%%%%%%%%%%%%%%%%%
%%%%%%%%%%%%%%%%%%%%%%%%%%%%%%%%%%%%%%%%%%%%%%%
%							ADJ ARRAY MODEL
%%%%%%%%%%%%%%%%%%%%%%%%%%%%%%%%%%%%%%%%%%%%%%%
%%%%%%%%%%%%%%%%%%%%%%%%%%%%%%%%%%%%%%%%%%%%%%%
%%%%%%%%%%%%%%%%%%%%%%%%%%%%%%%%%%%%%%%%%%%%%%%

\section{Cycle Detection in the Adjacency Array Model}\label{sec:adj_array}
Rather than an adjacency matrix, we may be provided with an adjacency array description of a graph as an input. In this model, the graph is given to us as a list of vertices associated to each vertex in the graph, which define its neighbours in the graph. Following \cite{durr2004quantum}, we assume that we are given the following information:

\begin{itemize}
\item The degrees of the vertices $d_1, d_2, ..., d_n$ and for every vertex $u$ an array with its neighbours $f_i : [d_i] \rightarrow [n]$. So $f_i(j)$ returns the $j^{\text{th}}$ neighbour of vertex $i$, according to some arbitrary but fixed numbering of the outgoing edges of $i$. We will assume that the input graph is undirected. 
\end{itemize}
Since we are given the degrees of each vertex for free, we can calculate the number of edges $m$ as $m = \frac{1}{2}\sum_{u \in V} d_u$. In this way, we can discard any $n$-vertex input graph with $m \geq n$, since such a graph must necessarily contain a cycle. This means that we only need to consider graphs with $m < n$. If we can map from the edges of $H$ to the edges of $G$ in the adjacency array model, then we can run a quantum walk on $H$, starting from $s$ and with $t$ as the single marked vertex, which, by a result of Belovs \cite{belovs_quantum_2013, belovs2013time}, can be used to decide s-t connectivity. By making use of the reduction of Lemma \ref{lem:reduction}, and the version of Grover search outlined in section \ref{sec:grover}, we can use a quantum walk in place of the span-program-based s-t connectivity algorithm to detect cycles in the adjacency array model in time $\tilde{O}(n\sqrt{d_m})$, where $d_m$ is the maximum degree of any vertex in the graph.

\subsection{Map from the edges of $H$ to the edges of $G$}
Given a vertex $u_b$ in $H$, we want to be able to produce an array of its neighbours. In particular, for a vertex $u_b$ we need a function $g_{u_b} : [d'_{u_b}] \rightarrow [3n]$, where $d'_{u_b}$ is the degree of vertex $u_b$ in $H$, so that $g_{u_b}(j)$ returns the $j^{th}$ neighbour of vertex $u_b$ in $H$. First, note that the degree of the vertex $u_b$ in $H$ is the same as the degree of the vertex $u$ in $G$ -- that is, $d'_{u_b} = d_u$, unless $u_b$ is connected to $s$ or $t$. Also, recall that the neighbours of vertex $u_b$ in $H$ are all of the form $v_{b'}$, where $v$ is a neighbour of $u$ in $G$, and $b'$ depends on the orientation of the edge $(u,v)$.

In general, suppose that we want to compute $g_{u_b}(j)$, the $j^{th}$ neighbour of vertex $u_b$. We know that it will be $v_{b'}$, for $v = f_u(j)$ (i.e. the $j^{th}$ neighbour of vertex $u$ in $G$) and some $b' \in \{0,1,2\}$. To calculate $b'$, we use the same process described in section \ref{sec:algo}, which begins by determining the direction of the edge $(u,v)$ as follows: if $v > u$, then the edge is directed $u \rightarrow v$, otherwise it is directed from $v \rightarrow u$. If $u = k$, then we look up the colour of vertex $v$ as determined by our hash function, and vice versa if $v = k$. If the colour is 0, we do nothing; if it is 1, we flip the edge. Finally, if the edge is directed $u \rightarrow v$, then $b' = b + 1 \mod 3$, otherwise $b' = b - 1 \mod 3$. We then return the answer: $g_{u_b}(j) = v_{b'}$. 

The functions $g_{u_b}$ for every vertex $u_b$ in $H$ can be computed using the function $f_{u}$ given by the adjacency array for vertex $u$, as well as some other operations that require $O(\text{polylog}(n))$ time and $O(\log n)$ space. Therefore, we can implement a quantum walk on the graph $H$ by implicitly querying the adjacency arrays for the vertices in $G$. The following section describes such a quantum walk.

\subsection{Quantum Walk for s-t Connectivity}
We use a quantum walk algorithm presented by Belovs in \cite{belovs_quantum_2013} and \cite{belovs2013time} for detecting a marked vertex in a graph, with the starting vertex set to $s$ and with $t$ being the only marked vertex in the graph. Then the presence of a path from $s$ to $t$ can be detected in $\tilde{O}(\sqrt{ln})$ steps of the quantum walk, where $l$ is the length of the path. Thus, given a vertex $k$ in the graph and some upper bound $d$ on the length of the cycle, we can detect the presence of a cycle that includes $k$ in $\tilde{O}(\sqrt{dn})$ steps of the quantum walk. Then by using the same variant of Grover search in section \ref{sec:grover}, we can detect the presence of an arbitrary cycle in $\tilde{O}(n)$ steps.

\subsection{Implementation}
In this section we describe an analogue of the algorithm from section \ref{sec:algo}, which uses Belovs' quantum walk in place of the span-program-based s-t connectivity algorithm of Reichardt and Belovs. As before, the algorithm takes as input a graph $G=(V,E)$ (except this time in the adjacency array model), a vertex $k \in V$, and some integer $d$. It outputs true with some probability when $G$ contains a cycle that includes $k$, and returns false with some probability when $G$ contains no cycles. Since we can dismiss graphs with more than $n$ edges in $\tilde{O}(n)$ time, we henceforth assume that the input graph has fewer than $n$ edges.

Once again, we consider the graph $H=(V',E')$ corresponding to the graph $G=(V,E)$. In order to apply the quantum walk, we must first make $H$ bipartite (which, in general, it will not be to begin with). To do this, we transform $H=(V',E')$ into $H'=(V'',E'')$ with vertex set $V'' = V' \times \{0,1\}$ and edge set $E'' = \{ ((u,0)(v,1), (u,0)(v,1)) : uv \in E'\}$. Then, we set $s = (k_0,0)$ and $t = (k_1,1)$. The graph is now bipartite, and we can still efficiently compute the neighbours of each vertex. Let $A$ be the set of vertices $\{(u,0) : u \in V'\}$ and $B$ the set of vertices $\{(u,1) : u \in V'\}$, and let $d_u$ denote the degree of vertex $u$.

The vectors $\{\ket{s} \otimes \ket{e_s}\} \cup \{\ket{u} \otimes \ket{v} : (u,v) \in E\}$ give the basis for the vector space of the quantum walk, which starts in the state $\ket{s}\ket{e_s}$. Let $\mathcal{H}_u = \text{span}(\{\ket{u}\ket{v} : (u,v) \in E\})$ denote the local space of vertex $u \neq s$, and $\mathcal{H}_s = \text{span} (\{\ket{s}\ket{v} : (s,v) \in E\} \cup \{\ket{s}\ket{e_s}\})$. A step of the quantum walk is defined as $R_AR_B$ where $R_A = \bigoplus_{u \in A}D_u$ and $R_B = \bigoplus_{u \in B}D_u$. Each diffusion operator $D_u$ acts only on $\mathcal{H}_u$, and is defined as follows:
\begin{itemize}
\item $D_t$ is the identity.
\item If $u \notin \{s,t\}$, then $D_u = I - 2\ket{\zeta_u}\bra{\zeta_u}$, where
\[
\ket{\zeta_u} = \frac{1}{\sqrt{d_u}}\sum_{(u,v) \in E} \ket{u}\ket{v}.
\]
\item $D_s = I - 2\ket{\zeta_s}\bra{\zeta_s}$, where
\begin{equation}\label{eq:local_s}
\ket{\zeta_s} =  \frac{1}{\sqrt{1 + d_sCd}} \left(\ket{s}\ket{e_s} + \sqrt{Cd}\sum_{(s,v) \in E} \ket{s}\ket{v}\right),
\end{equation}
for some constant $C$.
\end{itemize}

Then we have the following result, which follows directly from Theorem 4 of \cite{belovs_quantum_2013}:

\begin{theorem}\label{theo:belovs}
Given a graph $G=(V,E)$ such that $|E| \leq n$, two vertices $s$ and $t$ in $G$, and an integer $d$, then by applying $R_A$ and $R_B$ $O(\sqrt{dn})$ times we can detect a path from $s$ to $t$ with probability $\geq 2/3$ if a path of length $l \leq d$ exists, or otherwise say that no path exists with probability $\geq 2/3$.
\end{theorem}

By using the same arguments given in section \ref{sec:grover}, we can make use of the algorithm of Theorem \ref{theo:belovs} to detect arbitrary cycles by applying the operators $R_A$ and $R_B$ $\tilde{O}(n)$ times. Furthermore, we may use a special case of this algorithm to decide bipartiteness, also requiring $\tilde{O}(n)$ applications of $R_A$ and $R_B$. The efficiency of the algorithm then depends on the efficiency with which we can implement the reflections $R_A$ and $R_B$. 

\subsection{Implementing $R_A$ and $R_B$}\label{sec:reflection_ops}
We will restrict our attention to $R_A$; $R_B$ is implemented similarly (and is actually easier, since the vertex $s$, which requires a more complex diffusion operator, is in $A$). We need to implement
\begin{eqnarray*}
R_A = \bigoplus_{v \in A} D_v &=& \bigoplus_{v\in A}(I - 2\ket{\zeta_u}\bra{\zeta_u}) \\
&=& \bigoplus_{v\in A}(I - 2\ket{v}\bra{v} \otimes \ket{\phi_v}\bra{\phi_v}) \\
&=& I - 2\sum_{v\in A} (\ket{v}\bra{v} \otimes \ket{\phi_v}\bra{\phi_v})
\end{eqnarray*}
where we define $\ket{\phi_v} := \frac{1}{\sqrt{d_v}}\sum_{i \in [d_v]} \ket{f_v(i)}$ for $v \neq s$ (recall that we are given the degrees $d_1, d_2, ..., d_n$ of each vertex, and for each vertex $v$ a function $f_v: [d_v] \rightarrow [n]$, so that $f_v(j)$ returns the $j^{th}$ neighbour of vertex $v$.). $\ket{\phi_s}$ is defined slightly differently: there is an extra $\ket{s}$ term in equation (\ref{eq:local_s}), which we view as a `dangling' edge $e_s$ incident to vertex $s$, which is defined to be the $(d_s+1)^{th}$ neighbour of $s$, where $d_s$ is the degree of vertex $s$. That is, we add an entry to the adjacency array for $s$ so that $f_s(d_s+1) = e_s$. Then we may define
\[
\ket{\phi_s} := \frac{1}{\sqrt{1 + d_sCd}} \left(\ket{f_s(d_s+1)} + \sqrt{Cd}\sum_{i \in [d_s]} \ket{f_s(i)}\right).
\]
Intuitively, the $\ket{\phi_v}$ states represent the neighbours of the vertex $v$, in correspondence with the states $\ket{\zeta_v}$ given above. 

If we can implement a map $\ket{v}\ket{0} \mapsto \ket{v}\ket{\phi_v}$ for each vertex $v \in A$, then we may implement $R_A$ by performing the local reflections $I - 2\ket{v}\bra{v}\otimes \ket{\phi_v}\bra{\phi_v}$ in parallel for each $v\in A$. It is possible to implement the map efficiently for every vertex in superposition -- in particular, we have the following result:

\begin{restatable}{lemma}{lemreflect}\label{lem:reflect}
$R_A$ and $R_B$ can be implemented using $O(\sqrt{d_m})$ queries to the adjacency array, and $\text{\emph{polylog}}(n)$ additional operations per query.
\end{restatable}
The proof of this lemma is presented in Appendix \ref{app:diff_ops}. Since the quantum walk algorithm requires $\tilde{O}(n)$ applications of $R_AR_B$, then by Lemma \ref{lem:reflect} the total time required is $\tilde{O}(n\sqrt{d_m})$. If we are given no promise on the maximum degree of the graph, then in the worst case the algorithm will take time $\tilde{O}(n^{3/2})$, which matches the time complexity of the algorithm in the adjacency matrix model.

\subsection{Lower Bounds}
We provide $\Omega(n)$ quantum query lower bounds, which follow from almost the same reduction used by D{\" u}rr et al. in \cite{durr2004quantum} to prove a lower bound on s-t connectivity in the array model -- namely a reduction from the {\sc Parity} problem. The {\sc Parity} problem is defined as follows: given a bit-string $x \in \{0,1\}^p$ of length $p$, are there an even or an odd number of bits set to 1? Alternatively, we might consider the bit-string $x$ to be the output of some function for each of the input integers $0...p-1$. 

We reproduce the reduction here, and show how it leads to lower bounds for both cycle detection and bipartiteness.

\begin{figure}[htbp]
\begin{center}
\input{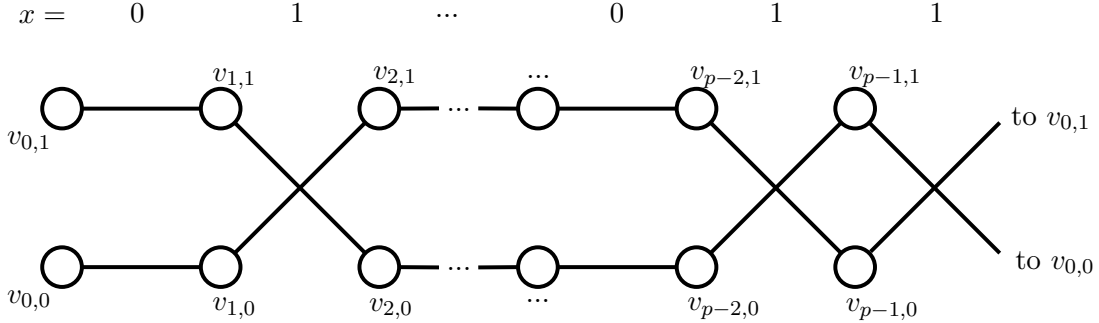}
\caption{Reduction from {\sc Parity} (similar to \cite{durr2004quantum})}
\label{fig:parity}
\end{center}
\end{figure}

\begin{lemma}
Bipartiteness testing and cycle detection both require $\Omega(n)$ queries in the adjacency array model.
\end{lemma}
\begin{proof}
Let $x \in \{0,1\}^p$ be an instance of the parity problem. We construct a permutation $f$ on $\{v_{i,b} : i \in [p], b \in \{0,1\}\}$ which has exactly 1 or 2 cycles, depending on the parity of $x$. We define $f(v_{i,b}) = v_{i+1,b\oplus x_i}$ and $f(v_{i+1,b\oplus x_i}) = v_{i,b}$ (so that $f$ is symmetric), where the addition is modulo $2p$ (and $\oplus$ denotes addition modulo 2). 

The (undirected) graph defined by $f$ has two levels and $p$ columns, each corresponding to a bit of $x$ -- see Figure \ref{fig:parity}. A walk starting at vertex $v_{0,0}$ and using each edge at most once, will go from left to right, changing level whenever the corresponding bit in $x$ is 1. So when $x$ is even, the walk returns to $v_{0,0}$ while having only explored half of the graph, otherwise it returns to $v_{0,1}$, and then connects from there to $v_{0,0}$ by $p$ more steps.
\\
\\
If we were to arbitrarily remove a single edge from the graph defined by $f$, we would either have no cycle present (if $x$ is odd), or exactly one cycle present (when $x$ is even). Therefore, if we can detect cycles in this modified graph, then we can decide the parity of $x$. That is, after removing a single (arbitrary) edge, there is a cycle present if and only if the parity of $x$ is even. This gives a $\Omega(n)$ quantum query lower bound for cycle detection in the adjacency array model. 
\\
\\
In the case of bipartiteness, we ensure that $x$ has an odd number of bits by adding a `dummy' bit $x_p = 0$ if $p = 2m$ for some integer $m$. We fix this dummy bit to zero, so that it doesn't affect the parity of $x$. This has the effect of adding two additional vertices $v_{p,0}$ and $v_{p,1}$ to the graph such that $f(v_{p,0}) = v_{0,0}$ and $f(v_{p,1}) = v_{1,1}$ (and vice versa). After this modification, we have a single cycle of length $2p+2$ in the graph if $x$ is even, or two disjoint cycles of length $p+1$ if $x$ is odd. Since $p+1$ is an odd integer, there is an odd cycle in the graph if and only if $x$ is odd.

In the case where $p$ is not even, we do not add the dummy bit, and so we also have an odd cycle in the graph if and only if $x$ is odd. In both cases, the graph is bipartite if and only if the parity of $x$ is odd. This gives the required bound.
\end{proof}

Note that these lower bounds are actually tight, since the quantum query complexity of s-t connectivity is $\Theta(n)$ in the array model \cite{durr2004quantum}, implying the existence of $O(n)$ quantum query algorithms for both bipartiteness testing and cycle detection, which can be obtained by applying an $O(n)$ query algorithm for s-t connectivity to the ancillary graph used in the reduction from cycle detection and bipartiteness.

\section{Acknowledgements}
CC was supported by the EPSRC. AM was supported by an EPSRC Early Career Fellowship (EP/L021005/1). AB was supported by the ERC Advanced Grant MQC.

%\begin{appendices}

%\end{appendices}
\appendix

\section{Span Programs}
\label{sec:span_programs_main}
Span programs are a linear algebraic model of computation, introduced by Karchmer and Wigderson in 1993 \cite{karchmer_span_1993}, that have many applications in classical complexity theory, and can be used to evaluate decision problems. Reichardt and \v{S}palek \cite{reichardt_span-program-based_2008} introduced a new complexity measure for span programs, the \emph{witness} size, which Reichardt later showed to have strong connections with quantum query complexity \cite{reichardt_span_2009, reichardt_reflections_2011}. In particular, he showed that the witness size of a span program and the query complexity of a quantum algorithm evaluating that span program are separated by at most a constant factor. This suggests that span programs may be useful for designing new quantum algorithms.

Span programs have been used to design quantum query algorithms for formula evaluation \cite{reichardt_span-program-based_2008, reichardt_span-program-based_2011,reichardt_faster_2011}, the matrix rank problem\cite{belovs_span-program-based_2011}, subgraph detection \cite{zhu_quantum_2012,belovs_span_2012-2,belovs_span_2012}, s-t connectivity \cite{belovs_span_2012}, and strong connectivity \cite{arins_span-program-based_2015}. For completeness, we briefly introduce this model; for further details, see \cite{karchmer_span_1993,belovs_span_2012}.

\subsection{Formal Definition}\label{sec:formal_sp}
A span program $\mathcal{P}$ takes as input an $n$-bit string $x \in \mathcal{D} \subseteq \{0,1\}^n$, and either accepts or rejects it. That is, it implements the (partial) boolean function $f_{\mathcal{P}} : \mathcal{D} \rightarrow \{0,1\}$. 

\begin{definition}
A span program is defined by a tuple $\mathcal{P} = (\mathcal{H}, \ket{\tau}, \{M_{i,b}\}, M_{\text{free}})$, where $\mathcal{H}$ is a finite-dimensional Hilbert space, $\ket{\tau} \in \mathcal{H}$ is the \emph{`target vector'}, $\{M_{i,b}\}$ is a set of sets of vectors for $i \in [n], b \in \{0,1\}$, where each $M_{i,b} \subseteq \mathcal{H}$ is a finite set of vectors which we will collectively call \emph{`input vectors'}, and $M_{\text{free}} \subseteq \mathcal{H}$ is a set of \emph{`free'} input vectors. \\
Given an input $x$, denote by $M(x) = \bigcup\{M_{i,b} : i \in [n], x_i = b\} \cup M_{\text{free}}$. Then the span program accepts if the target vector $\ket{\tau}$ can be written as a linear combination of the vectors in $M(x)$: \[ f(x) = 1 \iff \ket{\tau} \in \operatorname{span}(M(x)). \]
\end{definition}

Informally, the span program consists of sets of vectors that are either \emph{available} or \emph{unavailable}, depending on the input given to the span program. Generally speaking, we associate two sets of vectors to each input bit $x_i$, so that if $x_i = 1$, then the vectors in $M_{i,1}$ are available and those in $M_{i,0}$ are unavailable, and vice versa. The vectors contained in $M_{\text{free}}$ are always available, and any set $M_{i,b}$, $M_{\text{free}}$ may be empty.

\subsubsection{Witnesses}
\paragraph{Notation} Call $M(x)$ the set of available input vectors and let $d$ be the dimension of the Hilbert space $\mathcal{H}$ and $m$ be the total number of input vectors and free input vectors (also referred to as the `size' of $\mathcal{P}$). Write the set of all input vectors and free vectors as $\{\ket{v_j} : j \in [m]\}$. Finally, define $M := \sum_{j=1}^m \ket{v_j}\bra{j}$, which can be thought of as a matrix with all input vectors as columns.

\paragraph{Positive case} If $\mathcal{P}$ accepts $x$, then we can write $\ket{\tau}$ as a linear combination of available input vectors: \[\ket{\tau} = \sum_{v_j \in M(x)} w_j \ket{v_j} .\] Then the coefficients $w_j$ give a \emph{positive witness vector} for $x$, $\ket{w} = \sum_j w_j\ket{j}$, such that $M\ket{w} = \ket{\tau}$. The size of the witness is defined as $\|\ket{w}\|^2$. 

\paragraph{Negative case} If $\mathcal{P}$ rejects $x$, then it must not be possible to construct $\ket{\tau}$ using a linear combination of the available input vectors. Therefore, there must be some component of $\ket{\tau}$ that is orthogonal to all available input vectors. That is, there must exist some vector $\ket{w'}$ such that $\braket{w'|v_j} = 0$ for all $v_j \in M(x)$, and $\braket{w'|\tau} \neq 0$. In order for the witness size to be well defined, we require that $\braket{w'|\tau} = 1$. We call the vector $\ket{w'}$ the \emph{negative witness vector} for $x$.
The size of the witness is defined as 
\[
	\|M^\dag \ket{w'}\|^2 = \| \sum_{j=1}^m \ket{j}\braket{v_j|w'}\|^2 = \sum_{v_j \notin M(x)} |\braket{v_j|w'}|^2
\]
This equals the sum of the absolute squares of the inner products of $\ket{w'}$ with all \emph{unavailable} input vectors.

\paragraph{Witness size}
The witness size of $\mathcal{P}$ on input $x$, wsize($\mathcal{P}, x$), is defined as the minimum size among all witnesses for $x$. For domain $\mathcal{D} \subseteq \{0,1\}^n$, let \[ \text{wsize}_b(\mathcal{P}, \mathcal{D}) = \max_{x \in \mathcal{D}: f_\mathcal{P}(x) = b} \text{wsize}(\mathcal{P}, x) .\]
Then the witness size of $\mathcal{P}$ on domain $\mathcal{D}$ is defined as
\[
\text{wsize}(\mathcal{P}, \mathcal{D}) = \sqrt{\text{wsize}_0(\mathcal{P}, \mathcal{D}) \text{wisze}_1(\mathcal{P}, \mathcal{D})}
\]

%%%%%%%%%%%%%%%%%%%%%%%%%%%%%%%%%%%%%%%%%%%%%%%
%%%%%%%%%%%%%%%%%%%%%%%%%%%%%%%%%%%%%%%%%%%%%%%
%%%%%%%%%%%%%%%%%%%%%%%%%%%%%%%%%%%%%%%%%%%%%%%
%							ST_CONN SPAN PROGRAM
%%%%%%%%%%%%%%%%%%%%%%%%%%%%%%%%%%%%%%%%%%%%%%%
%%%%%%%%%%%%%%%%%%%%%%%%%%%%%%%%%%%%%%%%%%%%%%%
%%%%%%%%%%%%%%%%%%%%%%%%%%%%%%%%%%%%%%%%%%%%%%%
\section{Span Program for s-t Connectivity}\label{sec:stspan}
We present here a span program for solving the problem of s-t connectivity, due to Belovs and Reichardt \cite{belovs_span_2012}. Formally, the problem is defined as follows: given an $n$-vertex graph $G = (V,E)$, and two vertices $s, t \in V$, is there a path from $s$ to $t$ in $G$?

\subsection{Span Program}
Define a span program $\mathcal{P}$ using the vector space $\mathbb{R}^n$, with an orthonormal basis $ \{\ket{v} : v \in V\}$ -- i.e. a basis vector for each vertex in $G$. We suppose that the input to the program is a bit string of the form $x_{ij}$ for $i,j \in V$, such that $x_{ij} = 1$ iff there is an edge $(i,j) \in E$.
 Then the span program is defined as follows:
\begin{itemize}
\item \textbf{Target Vector:} $\ket{\tau} = \ket{t} - \ket{s}$.
\item \textbf{Available Input Vectors:} For each edge $(u,v) \in E$, (i.e. $x_{uv} = 1$), we make available the input vector $\ket{v} - \ket{u}$.
\end{itemize}
There are no free input vectors. It might be useful to note that in this example, the set of all input vectors is $\{\ket{j} - \ket{i} : i \neq j \in V\}$, with an input vector corresponding to every possible edge that might occur in $G$, given the vertex set $V$. Therefore, the total number of input vectors is $m = {n \choose 2}$.

Now we prove correctness and calculate the witness sizes.

\paragraph{Positive Case} Suppose $s$ and $t$ are connected in $G$. Then there must exist some path of length, say, $d$ between them: $s=u_0, u_1, ... , u_d=t$. Then all of the vectors $\ket{u_1}-\ket{s}, \ket{u_2}-\ket{u_1}, ... , \ket{t}-\ket{u_{d-1}}$ are available. Simply adding all these vectors, each with unit weight, gives $\ket{t} - \ket{s} = \ket{\tau}$. Since the positive witness will consist of $d$ entries of $+1$, the positive witness size is $O(d)$.

\paragraph{Negative Case} Suppose that $s$ and $t$ are not connected in $G$, and instead lie in different connected subcomponents of $G$. We must show that a negative witness $\ket{w'}$ exists, such that $\braket{w'|\tau} = 1$, and $\braket{w'|v} = 0$ for all available input vectors $v$. We can define $\ket{w'}$ by its inner product on all basis vectors: let $\braket{w'|u} = 1$ if $u$ is in the same connected subcomponent as $t$, and $\braket{w'|u} = 0$ otherwise. Then we have that $\braket{w'|\tau} = \bra{w'}(\ket{t} - \ket{s}) = \braket{w'|t} - \braket{w'|s} = 1 - 0 = 1$. If an input vector of the form $\ket{v} - \ket{u}$ is available, then there is an edge between vertices $u$ and $v$ in $G$, and therefore both $u$ and $v$ belong to the same connected subcomponent in $G$. Therefore, $\braket{w'|v} = \braket{w'|u}$ in all such cases, and $\ket{w'}$ is orthogonal to all available input vectors. Since there are at most ${n \choose 2}=O(n^2)$ unavailable input vectors, and the inner product between $\ket{w'}$ and any basis vector is either 0 or 1, the negative witness size is $O(n^2)$. 

The span program $\mathcal{P}$ therefore has witness size $O(n\sqrt{d})$.

%%%%%%%%%%%%%%%%%%%%%%%%%%%%%%%%%%%%%%%%%%%%%%%
%%%%%%%%%%%%%%%%%%%%%%%%%%%%%%%%%%%%%%%%%%%%%%%
%%%%%%%%%%%%%%%%%%%%%%%%%%%%%%%%%%%%%%%%%%%%%%%
%								TIME EFFICIENT
%%%%%%%%%%%%%%%%%%%%%%%%%%%%%%%%%%%%%%%%%%%%%%%
%%%%%%%%%%%%%%%%%%%%%%%%%%%%%%%%%%%%%%%%%%%%%%%
%%%%%%%%%%%%%%%%%%%%%%%%%%%%%%%%%%%%%%%%%%%%%%%
\section{Time Efficient Implementation of Span Program}\label{sec:impl_sp}
\subsection{Preliminaries}\label{sec:prelims}
We will require some facts about the eigenspaces of the product of two reflections. Let $A$ and $B$ be matrices each with $n$ rows and orthonormal columns. Let $\Pi_A = AA^\dag$ and $\Pi_B = BB^\dag$ be the projections onto the column spaces of $A$ and $B$, respectively. Let $R_A = 2\Pi_A - I$ and $R_B = 2\Pi_B - I$ be the reflections about the corresponding subspaces, and let $U = R_BR_A$ be their product. 

\begin{lemma}\label{lem:spec}\emph{(Spectral Lemma \cite{szegedy_quantum_2004})}. Under the above assumptions, all the singular values of $A^\dag B$ are at most 1. Let $\cos \theta_1, ..., \cos \theta_l$ be all the singular values of $A^\dag B$ lying in the open interval $(0,1)$, and let $\mathcal{C}(A)$ and $\mathcal{C}(B)$ denote the column spaces of $A$ and $B$, respectively. Then the following is a complete list of the eigenvalues of $U$:
\begin{itemize}
\item The $+1$ eigenspace is $(\mathcal{C}(A) \cap \mathcal{C}(B)) \oplus (\mathcal{C}(A)^\perp \cap \mathcal{C}(B)^\perp)$
\item The $-1$ eigenspace is $(\mathcal{C}(A) \cap \mathcal{C}(B)^\perp) \oplus (\mathcal{C}(A)^\perp \cap \mathcal{C}(B))$. Moreover, $\mathcal{C}(A)^\perp \cap \mathcal{C}(B) = B(\ker A^\dag B)$
\item On the orthogonal complement of the above subspaces, $U$ has eigenvalues $e^{2i\theta_j}$ and $e^{-2i\theta_j}$ for $j \in [l]$
\end{itemize}
\end{lemma}

\begin{lemma}[Effective Spectral Gap Lemma \cite{lee_quantum_2011}]\label{lem:eff_spec_gap} 
\textit{Let $P_\Theta$ be the orthogonal projection onto the span of all eigenvectors of $U$ with eigenvalues $e^{i\theta}$ such that $|\theta| \leq \Theta$. Then, for any vector $\ket{w}$ in the kernel of $\Pi_A$, we have
\begin{equation*}
\|P_\Theta\Pi_B\ket{w}\| \leq \frac{\Theta}{2}\|\ket{w}\|
\end{equation*}
} \\
\end{lemma}

We will also require the following tools, which have been used many times elsewhere in quantum algorithm design:
\begin{theorem}[Phase estimation \cite{kitaev_quantum_1995}\cite{cleve_quantum_1998}]\label{phase_estimation}
Given a unitary $U$ as a black box, there exists a quantum algorithm that, given an eigenvector $\ket{\psi}$ of $U$ with eigenvalue $e^{i\phi}$, outputs a real number $w$ such that $|w - \phi| \leq \delta$ with probability at least $9/10$. The algorithm uses $O(1/\delta)$ controlled applications of $U$ and $\frac{1}{\delta}\text{\emph{polylog}}(1/\delta)$ other elementary operations.
\end{theorem}
\begin{theorem}[Reflection using phase estimation \cite{magniez2011search}]\label{theo:phase_reflect}
Let $U\in U(n)$ have a unique eigenvector with eigenvalue 1, and let the smallest non-zero phase of $U$ be $\sigma_{\min}$. Then for any integer $k$ there exists a quantum circuit $R$ that acts on $O(\log_2 n) + ks$ qubits, where $s = \log_2\left(\frac{1}{\sigma_{\min}}\right) + O(1)$, such that:
\begin{itemize}
\item  $R$ uses the controlled-$U$ operator  $O(k2^{s})$ times and contains $O(ks^2)$ other gates.
\item If $\ket{\psi}$ is the unique $1$-eigenvector of $U$, then $R\ket{\psi}\ket{0^{ks}} = \ket{\psi}\ket{0^{ks}}$.
\item If $\ket{\phi}$ lies in the subspace orthogonal to $\ket{\psi}$, then $\|(R + I)\ket{\phi}\ket{0^{ks}}\| = O(1/2^{k})$.
\end{itemize}
\end{theorem}
The latter point of Theorem \ref{theo:phase_reflect} tells us that the circuit $R$ implements a reflection about the eigenvalue-$1$ eigenspace of $U$ up to some precision $2^{-k}$, determined by the value of $k$. In particular, if $U$ has a constant spectral gap, then $s = O(1)$ and the number of calls to the controlled-$U$ operator is $O(k)$, and depends only on our desired precision for the circuit. 
\\
\\

\subsection{Implementing Span Programs}
In this section we outline a general method, due to Belovs and Reichardt \cite{belovs_span_2012}, for implementing span programs in a time-efficient manner, and apply it to the span program for evaluating s-t connectivity. Before we do so, however, it will be useful to describe a quantum query algorithm for evaluating span programs and, in doing so, note an interesting connection between span programs and quantum query complexity:

\begin{theorem}[Reichardt \cite{reichardt_span_2009}]\label{theo:reichardt_query}
For any (partial) boolean function $f : \mathcal{D} \rightarrow \{0,1\}$, there is a quantum algorithm that requires $O(\text{wsize}(\mathcal{P},\mathcal{D}))$ queries to a quantum oracle for the bits of $x$, where $\mathcal{P}$ is any span program for which $f_\mathcal{P}$ agrees with f on the domain $\mathcal{D}$.
\end{theorem}
\begin{proof} \emph{(From \cite{belovs_span_2012}, included for completeness)}

The quantum algorithm works in the space $\mathbb{R}^{m+1}$ with orthonormal basis elements $\{\ket{j}\}^{m}_{j=0}$. Let $I(x) = \{j : v_j \in M(x)\}$ be the indices corresponding to the available input vectors, and let $W_0$, $W_1$ and $W=\sqrt{W_0W_1}$ be, respectively, the negative witness size, positive witness size, and witness size of $\mathcal{P}$. 

 We perform phase estimation on the operator $U = (2\Lambda - I)(2\Pi_x - I)$, the product of two reflections about the images of the projection operators $\Lambda$ and $\Pi_x$, which are defined as follows: Let $\Lambda : \mathbb{R}^{m+1} \rightarrow \mathbb{R}^{m+1}$ be the orthogonal projection onto the kernel of $\tilde{M}$, where:

\[
\tilde{M} = \frac{1}{\alpha}\ket{\tau}\bra{0} \oplus M = \frac{1}{\alpha}\ket{\tau}\bra{0} + \sum^{m}_{j=1}\ket{v_j}\bra{j} \text{ for some } \alpha \in \mathbb{R} \text{ yet to be defined,}
\]
\[
\text{and let } \Pi_x : \mathbb{R}^{m+1} \rightarrow \mathbb{R}^{m+1} \text{ be defined by } \Pi_x = \ket{0}\bra{0} + \sum_{j \in I(x)} \ket{j}\bra{j} 
\]
The algorithm accepts $x$ if and only if, on the input of $\ket{0}$, phase estimation on $U$, with precision $\Theta$, outputs a phase of zero (corresponding to an eigenvalue of 1). We need to find values of $\alpha$ and $\Theta$ for which this procedure works. 

First we consider the positive case. Take an optimal positive witness $\ket{w} = \sum_j w_j\ket{j}$, and use it to construct an eigenvalue 1 eigenvector of $U$: let $\ket{u} = \alpha\ket{0} - \ket{w}  = \alpha\ket{0} - \sum_j w_j\ket{j}$. Since $\ket{u}$ consists only of $\ket{0}$ and the positive witness vector, we have $\Pi_x\ket{u} = \ket{u}$ and $\tilde{M}\ket{u} = \ket{\tau} - \ket{\tau} = 0$. So $\ket{u}$ is in the kernel of $\tilde{M}$, which implies that $\Lambda\ket{u} = \ket{u}$. Thus, $\ket{u}$ is an eigenvalue 1 eigenvector of $U$, and phase estimation on $\ket{u}$ will output a phase of zero with certainty. Therefore, the probability of phase estimation on $\ket{0}$ outputting a phase of zero depends on the overlap of $\ket{0}$ with $\ket{u}$, and is at least:

\begin{equation}\label{eq:p_accept}
\frac{|\braket{0|u}|^2}{\|\ket{u}\|^2} = \frac{\alpha^2}{\alpha^2 + \sum w_j^2} \geq \frac{1}{1 + W_1 / \alpha^2}.
\end{equation}
  
Therefore, if we choose $\alpha = C\sqrt{W_1}$, for some constant $C$, we can increase the value of $C$ to make this probability arbitrarily close to 1. In particular, we can ensure that $CW > 1$
 
Now consider the negative case. Let $\ket{w'}$ be an optimal negative witness, and define $\ket{v} = \alpha\tilde{M}^\dag\ket{w'}$. Since $\ket{w'}$ is a negative witness, we have $\braket{w'|\tau} = 1$ and so $\ket{v} = \ket{0} + \alpha M^\dag\ket{w'}$. So $\Pi_x\ket{v} = \ket{0}$ and 

\begin{equation}\label{eq:p_reject1}
\|\ket{v}\|^2 \leq 1 + \alpha^2W_0  = 1 +  C^2W_1W_0 = 1 + C^2W^2 \leq 2C^2W^2 
\end{equation}
since the negative witness size $W_0$ is defined by $W_0 = \|M^\dag\ket{w'}\|^2$, and the final inequality follows from the restriction that $CW > 1$.

Let $\Theta$ be the precision of the phase estimation algorithm, and let $P_\Theta$ be the projection operators onto the space of eigenvectors of $U$ of phase less than $\Theta$. So the probability that phase estimation outputs a phase of zero on the input $\ket{0}$ is $\|P_\Theta\ket{0}\|^2 = \|P_\Theta\Pi_x\ket{v}\|^2$. 

Using the Effective Spectral Gap Lemma (Lemma \ref{lem:eff_spec_gap}), we have that $\|P_\Theta\Pi_x\ket{v}\| \leq \frac{\Theta}{2}\|\ket{v}\|$, as long as $\Lambda\ket{v} = 0$. This last condition is easy to verify, since $\ket{v}$ lies in the image of $\tilde{M}^\dag$. So, we have

\begin{equation}\label{eq:p_reject2}
\|P_\Theta\ket{0}\| = \|P_\Theta\Pi_x\ket{v}\| \leq \frac{\Theta}{2}\|\ket{v}\| \leq \Theta CW
\end{equation}
and we can choose the precision to be $\Theta = 1/C'W$, for some constant $C'$. Therefore, we can choose a large value of $C'$ such that the probability $\|P_\Theta\ket{0}\|^2$ is arbitrarily small. 

The phase estimation algorithm on $U$ with precision $\Theta$ requires $O(1/\Theta)$ queries to $U$, and each query to $U$ requires only one query to the oracle for $x$. Therefore the algorithm evaluates the span program on an input $x$ with $O(1/\Theta) = O(W)$ queries to $x$.
\end{proof}

The above algorithm defines a unitary operator $U = (2\Lambda - I)(2\Pi_x - I)$, which is the product of two reflections - the first, $R_\Pi := (2\Pi_x - I)$, is an input dependent reflection, and the second, $R_\Lambda := (2\Lambda - I)$, is an input independent reflection. Since the algorithm requires repeated applications of this operator, the algorithm may only be implemented time-efficiently if we can implement both $R_\Pi$ and $R_\Lambda$ time-efficiently. $R_\Pi$ is generally quite straightforward to implement - since all it requires is some efficient map from the bits of the input to the corresponding input vectors. However, the implementation of $R_\Lambda$ is more subtle, and will be the main focus of the rest of this section.

\subsubsection{Implementing $R_\Pi$ and $R_\Lambda$}\label{sec:general_imp}
We begin by describing a general approach for implementing the two reflections, which is due to Belovs and Reichardt \cite{belovs_span_2012}. 

We consider the $d \times (m+1)$ matrix $\tilde{M} = \frac{1}{\alpha}\ket{\tau}\bra{0} + \sum_{j=1}^m \ket{v_j}\bra{j}$ as the biadjacency matrix for a bipartite graph on $d + m+1$ vertices, and run a Szegedy-type quantum walk~\cite{szegedy_quantum_2004} on it. The structure of this graph is as follows: we have two disjoint sets of vertices -- one consisting of a vertex for every basis vector in our vector space, and one consisting of a vertex for the target vector plus every input vector in the span program. An edge exists between two vertices when a basis vector makes up some non-zero component of one or more of the input/target vectors. 

To perform a quantum walk, we must factor $M$ into two sets of unit vectors: vectors $\ket{a_i} \in \mathbb{R}^m$ for each row $i \in [d]$ and vectors $\ket{b_j} \in \mathbb{R}^d$ for each column $j \in [m]$, so that $\braket{i|b_j}\braket{a_i|j} = M'_{ij}$, where $M'$ differs from $M$ only by a rescaling of its rows, since rescaling the rows of a matrix does not affect its nullspace.

Given such a factorisation, let $A = \sum_{i=1}^d(\ket{i} \otimes \ket{a_i})\bra{i}$ and $B = \sum_{j=1}^m (\ket{b_j} \otimes \ket{j})\bra{j}$, so that $A^\dag B = M'$. Let $R_A$ and $R_B$ be the reflections about the column spaces of $A$ and $B$, respectively. 

Embed $\mathcal{H}$ into $\tilde{\mathcal{H}} = \mathbb{R}^d \otimes \mathbb{R}^m$ using the isometry $B$. Then $R_\Lambda$ can be implemented on $B(\mathcal{H})$ as the reflection about the $-1$ eigenspace of $R_BR_A$. By Lemma \ref{lem:spec}, this eigenspace equals $(\mathcal{C}(A) \cap \mathcal{C}(B)^\perp) \oplus (\mathcal{C}(A)^\perp \cap \mathcal{C}(B))$, which is equal to $B(\ker A^\dag B) = B(\ker M)$ plus a part that is orthogonal to $\mathcal{C}(B)$ and is therefore irrelevant. The reflection about the $-1$ eigenspace of $R_BR_A$ can then be implemented using phase estimation, which will give us a reflection about the kernel of $M$ in the larger space $\tilde{\mathcal{H}}$, whose basis is given by $\ket{i} \otimes \ket{j}$ for $i \in [d], j \in [m]$. $R_\Pi$ may be implemented by reflections controlled by $j$ - i.e. given a state in $\tilde{\mathcal{H}}$, multiply the phase by $-1$ if $\ket{v_j}$ is an unavailable input vector.

Intuitively, $A$ and $B$ are matrices that give us the \emph{local spaces} for each vertex $\ket{a_i}$ and $\ket{b_j}$, respectively. By local space, we are referring to the neighbours of a vertex in the graph described by $M$. Therefore, the reflection about the column spaces of $A$ and $B$ are equivalent to reflections about the local spaces $\ket{a_i}$ and $\ket{b_j}$, controlled by columns $i$ and $j$, respectively.

The efficiency of the algorithm thus depends on two factors:
\begin{enumerate}
\item The implementation costs of $R_A$ and $R_B$. Since these reflections decompose into local reflections, they can be easier to implement than $R_\Lambda$. 
\item The spectral gap around the $-1$ eigenvalue of $R_BR_A$, on which the efficiency of the phase estimation sub-routine will depend. By Lemma \ref{lem:spec}, this gap is determined by the gap of $A^\dag B = M'$ around singular value 0.
\end{enumerate}

Since, for any given span program, $M'$ describes a bipartite graph, the reflections $R_A$ and $R_B$ can usually be implemented efficiently. To calculate the properties of the spectral gap around singular value $0$ of $M'$, we may calculate the spectral gap around the eigenvalue $0$ of $\Delta := M'M'^\dag$ (since the singular values of $M'$ are the square roots of the eigenvalues of $\Delta$). 

We use phase estimation to perform the reflection about the $-1$ eigenspace of $R_BR_A$. If the smallest non-zero singular value of $M'$ is $\sigma_{\min}$, then by Theorem \ref{theo:phase_reflect} this will require $O((1/\sigma_{\min})\log(1/\delta))$ controlled applications of $R_A$ and $R_B$, plus $O(\log(1/\delta) \text{polylog}(1/\sigma_{\min}))$ other elementary operations, where $\delta$ is the precision of the circuit. Thus, if $R_A$ takes time $T_A$ and $R_B$ takes time $T_B$, the entire process will require time $\tilde{O}(\frac{T_A+T_B}{\sigma_{\min}})$ for constant $\delta$. If $M'$ has a constant spectral gap (i.e. $\sigma_{\min} = \Omega(1)$), the time required to implement the span program depends only on the complexity of the reflections $R_A$ and $R_B$.

\subsection{Implementation of the s-t connectivity span program}\label{app:stconn}
In this section we will apply the general approach described above to the span program for s-t connectivity as described in section \ref{sec:stspan}. Recall that we are given as input a graph $G=(V,E)$, and the indices of two vertices $s$ and $t$ from $V$, and we want to know whether or not there is a path from $s$ to $t$ in $G$. We will provide a proof for Theorem \ref{ther:belovs}, which we restate here for convenience:

\therbelovs*
%\begin{theorem}[Combination of Theorems 3 and 9 from \cite{belovs_span_2012}]%\label{ther:belovs}
%Consider the st-connectivity problem on a graph G given by its adjacency matrix. Assume there is a promise that if s and t are connected by a path, then they are connected by a path of length at most $d$. Then there exists a bounded-error algorithm that determines whether $s$ and $t$ are connected in $\tilde{O}(n\sqrt{d})$ time and uses $O(\log n)$ bits and qubits of storage, and which fails with probability at most $1/10$. 
%\end{theorem}
%\addtocounter{theorem}{-1}
%\end{group}

\begin{proof}
In order to make the implementation straightforward, we modify the span program from section \ref{sec:stspan} slightly. We assume that $s$ and $t$ are not directly connected by an edge -- a fact that can be checked in $O(n)$ time beforehand, if necessary. Then, alongside the normal scaled-down target vector $\ket{\tilde{\tau}} = \frac{1}{\alpha}(\ket{t} - \ket{s})$, we introduce a `never-available' input vector $\ket{\tilde{\sigma}} = \sqrt{1 - 1/\alpha^2}(\ket{t}-\ket{s})$. We may assume that $\alpha = C_1\sqrt{W_1} \geq 1$. To introduce a never-available input vector, we introduce a dummy input bit that is always set to zero, and associate the input vector with it. It is easy to verify that this modification changes neither the behaviour of the span program nor its witness size.
\\
\\
Define the set of input vectors (not including the one corresponding to an edge between $s$ and $t$)
\[
M_{in} := \{\ket{v_{xy}} := \ket{y}-\ket{x}: x \neq y \in V\} \setminus \{\ket{v_{st}}\},
\]
and let $\{\ket{xy} : \ket{v_{xy}} \in M_{in}\} \cup \{\ket{st}, \ket{\overline{st}}\}$ be an orthonormal basis for the set of indices of vectors in $M_{in}$, with the extra vectors $\ket{st}$ and $\ket{\overline{st}}$ indexing the scaled target vector $\ket{\tilde{\tau}}$ and the `never-available' input vector $\ket{\tilde{\sigma}}$, respectively.
Let
\[
\tilde{M} = \frac{1}{\alpha}\ket{\tau}\bra{st} + \sqrt{1 - 1/\alpha^2}\ket{\sigma}\bra{\overline{st}} + \sum_{\ket{v_{xy}} \in M_{in}} \ket{v_{xy}}\bra{xy}.
\]
Now we define some vectors $\ket{a_x}$ for each $x \in V$:
\begin{itemize}
\item For $x \notin \{s,t\}$, $\ket{a_x} = \frac{1}{\sqrt{n-1}} \sum_{y\in V\setminus \{x\}} \ket{xy}$
\item For $x \in \{s,t\}$, $\ket{a_x} = \frac{1}{\alpha\sqrt{n-1}}\ket{st} + \sqrt{\frac{1-1/\alpha^2}{n-1}}\ket{\overline{st}} + \frac{1}{\sqrt{n-1}}\sum_{y\in V\setminus \{s,t\}} \ket{xy}$
\end{itemize}
and some vectors $\ket{b_{ij}}$ for each input vector $\ket{v_{ij}} \in M_{in}$:
\begin{itemize}
\item $\ket{b_{ij}} = \frac{1}{\sqrt{2}}(\ket{j} - \ket{i})$
\end{itemize}

Then these $\ket{a_x}$ and $\ket{b_{ij}}$ give a factorisation of the matrix $\tilde{M}$, up to a rescaling of the rows. That is, for $x \notin \{s,t\}$,
\begin{eqnarray*}
\braket{a_x|ij}\braket{x|b_{ij}} &=& \frac{1}{\sqrt{2}\sqrt{n-1}} \left(\sum_{y\in V\setminus \{x\}} \braket{xy|ij}\right) \bra{x} (\ket{j} - \ket{i}) \\
&=& \frac{1}{\sqrt{2(n-1)}} \braket{x | v_{ij}}
\end{eqnarray*}
The cases $x\in\{s,t\}$ can be verified separately, and give the desired final result, implying that $M' = \frac{1}{\sqrt{2(n-1)}}M$.
\\
\\
Now we may define $A = \sum_{u \in V} (\ket{u} \otimes \ket{a_u})\bra{u}$ and $B = \sum_{u,v \in V}(\ket{b_{uv}} \otimes \ket{uv})\bra{uv}$, and proceed as in section \ref{sec:general_imp} -- i.e. we can now implement the reflection $R_\Lambda$ by using phase estimation to reflect about the -1 eigenspace of $R_BR_A$, where $R_B$ and $R_A$ are the reflections about $\mathcal{C}(B)$ and $\mathcal{C}(A)$, respectively. This reflection is independent of the input -- that is, we reflect about the null-space of the biadjacency matrix given by the basis vectors and the set of (all possible) input vectors, which is formally described above. Our approach is then to alternate reflections about this space and the space of all \emph{available} input vectors, with the latter being achieved by the reflection operator $R_\Pi$. We can implement $R_\Pi$ by querying the input graph, and multiplying by a phase of $-1$ if the edge corresponding to a given input vector is not present. 

\paragraph{Spectral gap of \bf{$M'$} --} 
Since we use phase estimation to implement the reflection about the $-1$ eigenspace of $R_BR_A$, the efficiency of the algorithm will depend upon the spectral gap around the $-1$ eigenvalue of $R_BR_A$. By Lemma \ref{lem:spec}, this gap is determined by the spectral gap around singular value 0 of $A^\dag B = M'$. The non-zero singular values of $M'$ are the square roots of the non-zero eigenvalues of $\Delta := M'M'^\dag$. Recall that $m=|M_{\text{in}}|$ gives the total number of input vectors in the span program, and that $M' = \frac{1}{\sqrt{2(n-1)}}M = \frac{1}{\sqrt{2(n-1)}}\sum_{i \in V}\sum_{j \in [m]} \braket{i|v_j} \ket{i}\bra{j}$, where each $\ket{v_j}$ corresponds to an `ordinary' input vector of the form $\ket{y}-\ket{x}$ for some $x\neq y \in V$, or to one of the special input vectors $\tilde{\tau} = \frac{1}{\alpha}(\ket{t} - \ket{s})$ or $\tilde{\sigma} = \sqrt{1 - \frac{1}{\alpha^2}}(\ket{t} - \ket{s})$. We have that 
\[
M'M'^\dag = \sum_{i, i' \in V} \left(\frac{1}{2(n-1)} \sum_{j \in [m]}\braket{i|v_j}\braket{v_j|i'}\right)\ket{i}\bra{i'},
\]
and therefore we can compute $\Delta$ by inspecting the individual values $M'M'^\dag_{ii'}$ for different cases of $i$ and $i'$:

\begin{itemize}

\item If $i \notin \{s,t\}$, and/or $i' \notin \{s,t\}$, then we consider two cases:
\begin{enumerate}
\item $i = i'$: In this case, the term inside the brackets contributes a value of $\frac{n-1}{2(n-1)} = \frac{1}{2}$ to $\Delta_{ii'}$, since each vertex $i$ has degree $(n-1)$. Alternatively, we note that each basis vector has $(n-1)$ input vectors with which it has inner product 1, and thus $\sum_j |\braket{i|v_j}|^2 = (n-1)$.
\item $i \neq i'$: In this case, there is exactly one edge between vertices $i$ and $i'$, which contributes a term of $-\frac{1}{2(n-1)}$ to $\Delta_{ii'}$.
\end{enumerate}

\item If $i,i' \in \{s,t\}$, the situation is slightly different, but the result is the same. Again, we will deal with two cases:
\begin{enumerate}
\item $i = i'$: In this case, vertex $i$ has $n$ neighbours. $(n-2)$ of them correspond to ordinary edges, whilst the remaining two correspond to the scaled target vector $\ket{\tilde{\tau}}$, and the never-available input vector $\ket{\tilde{\sigma}}$. Then $|\braket{i|\tilde{\tau}}|^2 = 1/\alpha^2$ and $|\braket{i|\tilde{\sigma}}|^2 = (1 - 1/\alpha^2)$. So the term inside the brackets contributes a value of $(1 + (n-2) + 1/\alpha^2 - 1/\alpha^2)/2(n-1) = 1/2$ to $\Delta_{ii}$.
\item $i \neq i'$: In this case, there are two edges between vertices $i$ and $i'$ (since one must be $s$, and the other $t$). The first edge, $\ket{\tilde{\tau}}$, contributes a value of $-1/\alpha^2$, and the other, $\ket{\tilde{\sigma}}$, contributes a value of $(1/\alpha^2 - 1)$. Together, they contribute a term of $-\frac{1}{2(n-1)}$ to $\Delta_{ii'}$. 
\end{enumerate}

\end{itemize}

The result is an $n\times n$ square matrix, whose diagonal elements are $1/2$, and off-diagonal elements are $-1/2(n-1)$. By taking out a factor of $1/2(n-1)$, we obtain the Laplacian matrix for the complete graph on $n$ vertices, which has the form $nI_n - J_n$, where $I_n$ is the $n\times n$ identity matrix and $J_n$ the $n\times n$ all-ones matrix. This Laplacian has a single eigenvalue of 0, and $(n-1)$ eigenvalues of $n$. Therefore, the eigenvalues of $\Delta$ are 0 (multiplicity 1) and $n/2(n-1)$ (multiplicity $(n-1)$), and thus the non-zero singular values of $M'$ are all at least $1/\sqrt{2}$, giving us a constant spectral gap as desired. 

\paragraph{Implementing $R_A$ and $R_B$ -- }
Now it remains to show that we can implement the reflection operators $R_A$ and $R_B$ efficiently. Since these reflections decompose into local reflections, it suffices to describe operators that implement local reflections about each $\ket{a_u}$ and $\ket{b_{uv}}$. 

Recall that the algorithm works in the Hilbert space spanned by the vectors $\ket{i} \otimes \ket{j}$, where $i$ varies over the vertices in the graph, and $j$ over the input vectors and the target vector. We will describe the implementation of $R_A$ first. For all $u \notin \{s,t\}$, $\ket{u} \otimes \ket{a_u}$ is the uniform superposition of the states $\{ \ket{u} \otimes \ket{uv} : v \in V\setminus \{u\}\}$, and so the reflection is a Grover diffusion operator. 

For $u = s$, the transformation is slightly more complex. Let $F$ be the Fourier transform on the space spanned by $\{\ket{s}\otimes \ket{sv} : v \in V\setminus\{s\}\}$ that maps $\ket{s}\otimes\ket{st}$ to the uniform superposition; let $K$ be a unitary on the space spanned by $\{\ket{s}\otimes\ket{st}, \ket{s}\otimes\ket{\overline{st}}\}$ that maps $\ket{s}\otimes\ket{st}$ to $\frac{1}{\alpha}(\ket{s}\otimes\ket{st}) + \sqrt{1 - \frac{1}{\alpha^2}}(\ket{s}\otimes\ket{\overline{st}})$; and let $L$ be a unitary that multiplies the phase of all states except $\ket{s}\otimes\ket{st}$ by -1. Then the local reflection can be implemented by $FKLK^{-1}F^{-1}$. Intuitively, this is still similar to the Grover diffusion operator: the unitary $K$ acts to spread out the amplitude on the $st$ edge between the target vector $\ket{\tilde{\tau}}$ and the input vector $\ket{\tilde{\sigma}}$, and when combined with $F$ it creates the desired superposition over edges adjacent to $s$. Finally, the unitary $L$ performs the reflection, analogously to a standard diffusion operator. A very similar operation works for $u=t$. 

The implementation of $R_B$ is relatively straightforward. We apply the negated swap to all pairs $(\ket{u}\otimes\ket{uv}, \ket{v}\otimes\ket{uv})$, which maps a pair $(\ket{x}, \ket{y}) \mapsto (-\ket{y}, -\ket{x})$, for some arbitrary states $\ket{x}, \ket{y}$. This can be achieved in logarithmic time. 

To implement $R_\Pi$, the algorithm checks, for each state $\ket{i}\otimes\ket{j}$, whether input vector $j$ is available by querying the input oracle for the presence of edge $j$. If it is available, it does nothing, otherwise it negates the phase of the state. 

The states $\ket{i}\otimes\ket{j}$ can be stored using a logarithmic number of qubits. In particular, for a graph with $n$ vertices, we require $O(\log n)$ qubits.

\end{proof}

There may be cases where we do not know an upper bound on the length of the path between $s$ and $t$ ahead of time. In such situations, we may want to `guess' an upper bound on the length of the path. Provided that our guess is only wrong by at most a constant factor, the s-t connectivity algorithm of Theorem \ref{ther:belovs} will still fail with probability at most $1/10$, which follows directly from the statement of Theorem \ref{ther:belovs}. In the case that our guess is smaller than the actual length of the path (by more than a constant factor), then the algorithm may fail with a high probability.

\section{Diffusion Operators for Quantum Walk}\label{app:diff_ops}
Here we give details for implementing the diffusion operators of the quantum walk used in Section \ref{sec:reflection_ops}. In particular, we provide a proof for Lemma \ref{lem:reflect}, which we restate below.
\lemreflect*
\begin{proof}

We can query the adjacency array of each vertex in superposition. In particular, let 
\[
\ket{\psi_v} = \frac{1}{\sqrt{d_v}}\sum_{i \in [d_v]}\ket{i}\ket{f_v(i)}
\]
be the state that results from querying the adjacency array of vertex $v$ in superposition. Recall that we define $\ket{\phi_v} := \sum_{i \in [d_v]}\ket{f_v(i)}$. We want to produce $\ket{\phi_v}$ from the state $\ket{\psi_v}$. 

Let $\ket{+_v} := \frac{1}{\sqrt{d_v}} \sum_{i \in [d_v]}\ket{i}$. If we perform the $\{\ket{+_v}\bra{+_v}, I-\ket{+_v}\bra{+_v}\}$ measurement on the first register, then the first outcome is obtained with probability $1/d_v$, and in this case the second register collapses to $\ket{\phi_v}$. We want to maximise the probability of measuring $\ket{+_v}$ in the first register, in order to produce the desired state in the second register.

In fact, we can increase the probability of measuring $\ket{+_v}$ to certainty using exact amplitude amplification, which also gives us the desired state in the second register. Define 
\[
\ket{\Phi_v} := \ket{+_v} \ket{\phi_v},
\]
and two projectors
\[
P_+ := \ket{+_v}\bra{+_v} \otimes I
\]
and
\[
P_\psi := \ket{\psi_v}\bra{\psi_v}.
\]
We see that 
\begin{eqnarray*}
P_+\ket{\psi_v} &=& \frac{1}{\sqrt{d_v}}\sum_{i \in [d_v]}\left( \ket{+_v}\braket{+_v|i} \otimes \ket{f_v(i)} \right) \\
&=& \frac{1}{\sqrt{d_v}}\sum_{i \in [d_v]}\left( \frac{1}{\sqrt{d_v}}\ket{+_v} \otimes \ket{f_v(i)} \right) \\
&=& \frac{1}{\sqrt{d_v}}\ket{+_v}\otimes \frac{1}{\sqrt{d_v}}\sum_{i \in [d_v]}\ket{f_v(i)} \\
&=& \frac{1}{\sqrt{d_v}}\ket{\Phi_v} \\
&=& \ket{\Phi_v}\braket{\Phi_v|\psi_v}
\end{eqnarray*}
and
\begin{eqnarray*}
P_+\ket{\Phi_v} &=& \ket{+_v}\braket{+_v|+_v} \otimes \frac{1}{\sqrt{d_v}}\sum_{i \in [d_v]}\ket{f_v(i)} \\
&=& \ket{+_v}\otimes \frac{1}{\sqrt{d_v}}\sum_{i \in[d_v]}\ket{f_v(i)} \\
&=& \ket{+_v}\otimes \ket{\phi_v}  = \ket{\Phi_v}\\
&=& \ket{\Phi_v}\braket{\Phi_v|\Phi_v}
\end{eqnarray*}
That is, in the subspace spanned by $\{\ket{\psi_v}, \ket{\Phi_v}\}$, the projector $P_+$ acts as the projector $\ket{\Phi_v}\bra{\Phi_v}$. We can define two operators $R_+ = I - 2P_+$ and $R_\psi = I - 2P_\psi$, which, taking into account the observation noted above, are inversions about the spaces spanned by $\ket{\Phi_v}$ and $\ket{\psi_v}$, respectively (within the subspace spanned by $\{\ket{\psi_v}, \ket{\Phi_v}\}$). Thus, by alternating the two reflections, we can use the exact variant of amplitude amplification in the standard way to produce the state $\ket{\Phi_v}$ from the state $\ket{\psi_v}$. In particular, we apply $Q^m := (R_\psi R_+)^m$ to the initial state $\ket{\psi_v}$ for some integer $m$, followed by one application of a modified version of $Q$ which performs smaller rotations in order to make the algorithm exact. Since $Q$ preserves the subspace spanned by $\{\ket{\psi_v}, \ket{\Phi_v}\}$, then (by the arguments above) the algorithm will produce the desired state $\ket{\Phi_v}$. We have that $|\braket{\psi_v|\Phi_v}|^2 = \frac{1}{d_v}$, and so we can choose an $m = \Theta(\sqrt{d_v})$ to obtain the state $\ket{\Phi_v}$ in $\Theta(\sqrt{d_v})$ time \cite{brassard2002quantum}.
\\
\\
In other words, we can produce a state $\ket{v}\ket{+_v}\ket{\phi_v}$ from an initial state $\ket{v}\ket{0}\ket{0}$ in $\Theta(\sqrt{d_v})$ time. We can then uncompute the value in the second register, giving us $\ket{v}\ket{\phi_v}\ket{0}$ (where we have swapped the final two registers for clarity). Let $U$ be the operator that maps the state $\ket{v}\ket{0}\ket{0}$ to $\ket{v}\ket{\phi_v}\ket{0}$. Then the diffusion operator $D_v$ may be implemented by $US_0U^{-1}$, where $S_0$ changes the sign of the amplitude if and only if the final two registers are in the all zero state. That is, it performs the map $S_0\ket{v}\ket{0}\ket{0} \mapsto -\ket{v}\ket{0}\ket{0}$ for all $v \in A$. More precisely, it implements the reflection
\[
S_0 = I - 2\sum_{v\in A}\ket{v}\bra{v} \otimes \ket{0}\bra{0} \otimes \ket{0}\bra{0}.
\]
Therefore,
\begin{eqnarray*}
US_0U^{-1} &=& U( I - 2\sum_{v\in A}\ket{v}\bra{v} \otimes \ket{0}\bra{0} \otimes \ket{0}\bra{0})U^{-1} \\
&=& I - 2\sum_{v\in A} \ket{v}\bra{v} \otimes \ket{\phi_v}\bra{\phi_v} \otimes \ket{0}\bra{0} \\
&=& R_A
\end{eqnarray*}
by the definition of $R_A$.
\\
\\
Since we are applying $U$ to all vertices in superposition, it will be necessary to clarify how the amplitude amplification part of $U$ can be applied to a superposition over vertices. In order to produce the desired state for some vertex $v$, amplitude amplification needs to be performed for a number of iterations that depends upon the degree of that vertex. Since different vertices will have different degrees, the number of iterations required will vary between vertices. To address this, we can define an algorithm $AA(v, d_v, d_m)$, where $d_m$ is the maximum degree of any vertex. The algorithm will perform amplitude amplification for the correct number of iterations for vertex $v$, and then will do nothing (i.e. apply the identity operator) for the remaining iterations. In this way, the amplitude amplification routines stop and wait for the vertex with the largest degree, and thus the amplitude amplification step can be applied to all vertices in superposition, requiring time $O(\sqrt{d_m})$.

Since the diffusion operators required for the vertices $s$ and $t$ are different, it is worth discussing their implementations separately. The diffusion operator for vertex $t$ is easy to implement, since it is the identity. Vertex $s$ has a more complicated operator; however, all we need to change is the operation that maps the state $\ket{s}\ket{0}\ket{0}$ to the state $\ket{s}\ket{\psi_s}$. For all $v \notin \{s,t\}$, we simply produce a uniform superposition over the neighbours of $v$. For $s$, as discussed above, we have an additional term corresponding to an additional edge incident to vertex $s$, and therefore we need to be able to produce the state 
\[
\ket{+_s} = \frac{1}{\sqrt{1 + d_sCd}}\ket{d_s} + \sqrt{\frac{Cd}{1 + d_sCd}} \sum_{i=0}^{d_s-1} \ket{i}
\]
in the second register, which will be used to produce $\ket{\phi_s}$. In order to do this, let $K$ be a unitary operator on the space spanned by $\{\ket{0}, \ket{d_s}\}$ that maps $\ket{0}$ to $\frac{1}{\sqrt{1+d_sCd}}\ket{d_s} + \sqrt{\frac{d_sCd}{1+d_sCd}}\ket{0}$. Then let $F$ be the Fourier transform on the space spanned by $\{\ket{i} : i \in \{0..d_s-1\}\}$ that maps $\ket{0}$ to the uniform superposition. Then the required state can be produced by applying $FK$ to the state $\ket{0}$. We can then proceed as in the more general case to implement the local reflection.
\\
\\
In general, the time taken to implement the operators $R_A$ and $R_B$ will depend on the degrees of the vertices in the graph. In particular, if the maximum degree of any one vertex is $d_m$, then we will have to use at most $O(\sqrt{d_m})$ iterations of amplitude amplification in order to implement the local reflections in parallel. Therefore $O(\sqrt{d_m})$ queries are required to implement $R_AR_B$, and the time complexity is the same up to $\text{polylog}$ factors in $n$. 
\end{proof}

\bibliographystyle{abbrv}
\bibliography{references.bib}

\begin{thebibliography}{10}

\bibitem{aleliunas_random_1979}
R.~Aleliunas, R.~Karp, R.~Lipton, L.~Lovasz, and C.~Rackoff.
\newblock Random walks, universal traversal sequences, and the complexity of
  maze problems.
\newblock In {\em 20th {Annual} {Symposium} on {Foundations} of {Computer}
  {Science}}, pages 218--223, 1979.

\bibitem{ambainis2008quantum}
A.~Ambainis, K.~Iwama, M.~Nakanishi, H.~Nishimura, R.~Raymond, S.~Tani, and
  S.~Yamashita.
\newblock Quantum query complexity of boolean functions with small on-sets.
\newblock In {\em International Symposium on Algorithms and Computation}, pages
  907--918. Springer, 2008.

\bibitem{arins_span-program-based_2015}
A.~\={A}ri\c{n}\v{s}.
\newblock Span-program-based quantum algorithms for graph bipartiteness and
  connectivity.
\newblock 2015.
\newblock {\tt arXiv:1510.07825}.

\bibitem{belovs_span-program-based_2011}
A.~Belovs.
\newblock Span-program-based quantum algorithm for the rank problem.
\newblock 2011.
\newblock {\tt arXiv:1103.0842}.

\bibitem{belovs_span_2012-2}
A.~Belovs.
\newblock Span {Programs} for {functions} with {constant}-sized 1-certificates.
\newblock In {\em Proceedings of the {Forty}-fourth {Annual} {ACM} {Symposium}
  on {Theory} of {Computing}}, {STOC} '12, pages 77--84, New York, NY, USA,
  2012. ACM.
\newblock {\tt arXiv:1105.4024}.

\bibitem{belovs_quantum_2013}
A.~Belovs.
\newblock Quantum {walks} and {electric} {networks}.
\newblock 2013.
\newblock {\tt arXiv:1302.3143}.

\bibitem{belovs2013time}
A.~Belovs, A.~M. Childs, S.~Jeffery, R.~Kothari, and F.~Magniez.
\newblock Time-efficient quantum walks for 3-distinctness.
\newblock In {\em International Colloquium on Automata, Languages, and
  Programming}, pages 105--122. Springer, 2013.

\bibitem{belovs_span_2012}
A.~Belovs and B.~Reichardt.
\newblock Span {Programs} and {quantum} {algorithms} for st-{connectivity} and
  {claw} {detection}.
\newblock In {\em European Symposium on Algorithms ({ESA}) 2012}, number 7501
  in Lecture {Notes} in {Computer} {Science}, pages 193--204. Springer Berlin
  Heidelberg, 2012.
\newblock {\tt arXiv:1203.2603}.

\bibitem{bernstein1997quantum}
E.~Bernstein and U.~Vazirani.
\newblock Quantum complexity theory.
\newblock {\em SIAM Journal on Computing}, 26(5):1411--1473, 1997.

\bibitem{boyer1996tight}
M.~Boyer, G.~Brassard, P.~H{\o}yer, and A.~Tapp.
\newblock Tight bounds on quantum searching.
\newblock {\em Fortschritte der Physik}, 46(4-5):493--506, 1998.
\newblock {\tt arXiv:quant-ph/9605034}.

\bibitem{brassard2002quantum}
G.~Brassard, P.~H{\o}yer, M.~Mosca, and A.~Tapp.
\newblock Quantum amplitude amplification and estimation.
\newblock 2000.
\newblock {\tt arXiv:quant-phi/0401091}.

\bibitem{childs_quantum_2012}
A.~Childs and R.~Kothari.
\newblock Quantum {query} {complexity} of {minor}-{closed} {graph}
  {properties}.
\newblock {\em SIAM Journal on Computing}, 41(6):1426--1450, 2012.
\newblock {\tt arXiv:1011.1443}.

\bibitem{cleve_quantum_1998}
R.~Cleve, A.~Ekert, C.~Macchiavello, and M.~Mosca.
\newblock Quantum algorithms revisited.
\newblock {\em Proceedings of the Royal Society of London A: Mathematical,
  Physical and Engineering Sciences}, 454(1969):339--354, 1998.
\newblock {\tt arXiv:quant-ph/9708016}.

\bibitem{durr2004quantum}
C.~D{\"u}rr, M.~Heiligman, P.~H{\o}yer, and M.~Mhalla.
\newblock Quantum query complexity of some graph problems.
\newblock In {\em Automata, Languages and Programming}, pages 481--493.
  Springer, 2004.
\newblock {\tt arXiv:quant-ph/0401091}.

\bibitem{giovannetti2008quantum}
V.~Giovannetti, S.~Lloyd, and L.~Maccone.
\newblock Quantum random access memory.
\newblock {\em Physical review letters}, 100(16):160501, 2008.
\newblock {\tt arXiv:0708.1879}.

\bibitem{hoyer2003quantum}
P.~H{\o}yer, M.~Mosca, and R.~de~Wolf.
\newblock Quantum search on bounded-error inputs.
\newblock In {\em Automata, Languages and Programming}, pages 291--299.
  Springer, 2003.

\bibitem{jeffery2015nand}
S.~Jeffery and S.~Kimmel.
\newblock Nand-trees, average choice complexity, and effective resistance.
\newblock {\em {\tt arXiv:1511.02235}}, 2015.

\bibitem{karchmer_span_1993}
M.~Karchmer and A.~Wigderson.
\newblock On {Span} {Programs}.
\newblock In {\em Structure in {Complexity} {Theory} {Conference}}, pages
  102--111, 1993.

\bibitem{kitaev_quantum_1995}
A.~Y. Kitaev.
\newblock Quantum measurements and the {Abelian} {Stabilizer} {Problem}.
\newblock 1995.
\newblock {\tt arXiv:quant-ph/9511026}.

\bibitem{le2014improved}
F.~Le~Gall.
\newblock Improved quantum algorithm for triangle finding via combinatorial
  arguments.
\newblock In {\em Foundations of Computer Science (FOCS), 2014 IEEE 55th Annual
  Symposium on}, pages 216--225. IEEE, 2014.
\newblock {\tt arXiv:1407.0085}.

\bibitem{lee_quantum_2011}
T.~Lee, R.~Mittal, B.~Reichardt, R.~{\v{S}}palek, and M.~Szegedy.
\newblock Quantum {query} {complexity} of {state} {conversion}.
\newblock In {\em 2011 {IEEE} 52nd {Annual} {Symposium} on {Foundations} of
  {Computer} {Science} ({FOCS})}, pages 344--353, 2011.
\newblock {\tt arXiv:1011.3020}.

\bibitem{magniez2011search}
F.~Magniez, A.~Nayak, J.~Roland, and M.~Santha.
\newblock Search via quantum walk.
\newblock {\em SIAM Journal on Computing}, 40(1):142--164, 2011.
\newblock {\tt arXiv:quant-ph/0608026}.

\bibitem{mitzenmacher2005probability}
M.~Mitzenmacher and E.~Upfal.
\newblock {\em Probability and computing: Randomized algorithms and
  probabilistic analysis}.
\newblock Cambridge University Press, 2005.

\bibitem{piddock}
S.~Piddock.
\newblock A quantum algorithm for detecting cycles using span programs.
\newblock In preparation.

\bibitem{reichardt_span_2009}
B.~Reichardt.
\newblock Span programs and quantum query complexity: The general adversary
  bound is nearly tight for every boolean function.
\newblock In {\em Foundations of Computer Science, 2009. FOCS'09. 50th Annual
  IEEE Symposium on}, pages 544--551. IEEE, 2009.
\newblock {\tt arXiv:0904.2759}.

\bibitem{reichardt_faster_2011}
B.~Reichardt.
\newblock Faster {quantum} {algorithm} for {evaluating} {game} {trees}.
\newblock In {\em Proceedings of the {Twenty}-second {Annual} {ACM}-{SIAM}
  {Symposium} on {Discrete} {Algorithms}}, {SODA} '11, pages 546--559, San
  Francisco, California, 2011. SIAM.
\newblock {\tt arXiv:0907.1623}.

\bibitem{reichardt_reflections_2011}
B.~Reichardt.
\newblock Reflections for {quantum} {query} {algorithms}.
\newblock In {\em Proceedings of the {Twenty}-second {Annual} {ACM}-{SIAM}
  {Symposium} on {Discrete} {Algorithms}}, {SODA} '11, pages 560--569, San
  Francisco, California, 2011. SIAM.
\newblock {\tt arXiv:1005.1601}.

\bibitem{reichardt_span-program-based_2011}
B.~Reichardt.
\newblock Span-{Program}-{based} {quantum} {algorithm} for {evaluating}
  {unbalanced} {formulas}.
\newblock In {\em Theory of {Quantum} {Computation}, {Communication}, and
  {Cryptography}}, number 6745 in Lecture {Notes} in {Computer} {Science},
  pages 73--103. Springer Berlin Heidelberg, 2011.
\newblock {\tt arXiv:0907.1622}.

\bibitem{reichardt_span-program-based_2008}
B.~Reichardt and R.~{\v{S}}palek.
\newblock Span-{Program}-based {quantum} {algorithm} for {evaluating}
  {formulas}.
\newblock In {\em Proceedings of the {Fortieth} {Annual} {ACM} {Symposium} on
  {Theory} of {Computing}}, {STOC} '08, pages 103--112, New York, NY, USA,
  2008. ACM.
\newblock {\tt arXiv:0710.2630}.

\bibitem{robertson_graph_2004}
N.~Robertson and P.~D. Seymour.
\newblock Graph {minors}. {XX}. {Wagner}'s conjecture.
\newblock {\em Journal of Combinatorial Theory, Series B}, 92(2):325--357,
  2004.

\bibitem{rosenberg1973time}
A.~L. Rosenberg.
\newblock On the time required to recognize properties of graphs: A problem.
\newblock {\em ACM SIGACT News}, 5(4):15--16, 1973.

\bibitem{sun2004graph}
X.~Sun, A.~C. Yao, and S.~Zhang.
\newblock Graph properties and circular functions: How low can quantum query
  complexity go?
\newblock In {\em Computational Complexity, 2004. Proceedings. 19th IEEE Annual
  Conference on}, pages 286--293. IEEE, 2004.

\bibitem{szegedy_quantum_2004}
M.~Szegedy.
\newblock Quantum speed-up of {Markov} chain based algorithms.
\newblock In {\em 45th {Annual} {IEEE} {Symposium} on {Foundations} of
  {Computer} {Science}, 2004. {Proceedings}}, pages 32--41, 2004.
\newblock {\tt arXiv:quant-ph/0401053}.

\bibitem{wang_span-program-based_2013}
G.~Wang.
\newblock Span-program-based quantum algorithm for tree detection.
\newblock Sept. 2013.
\newblock {\tt arXiv:1309.7713}.

\bibitem{watrous1999quantum}
J.~Watrous.
\newblock Quantum simulations of classical random walks and undirected graph
  connectivity.
\newblock In {\em Computational Complexity, 1999. Proceedings. Fourteenth
  Annual IEEE Conference on}, pages 180--187. IEEE, 1999.
\newblock {\tt arXiv:cs/9812012}.

\bibitem{watrous1999space}
J.~Watrous.
\newblock Space-bounded quantum complexity.
\newblock {\em Journal of Computer and System Sciences}, 59(2):281--326, 1999.

\bibitem{zhang2005power}
S.~Zhang.
\newblock On the power of {A}mbainis lower bounds.
\newblock {\em Theoretical Computer Science}, 339(2):241--256, 2005.
\newblock {\tt arXiv:quant-ph/0311060}.

\bibitem{zhu_quantum_2012}
Y.~Zhu.
\newblock Quantum query complexity of constant-sized subgraph containment.
\newblock {\em International Journal of Quantum Information}, 10(03):1250019,
  2012.
\newblock {\tt arXiv:1109.4165}.

\end{thebibliography}

\end{document}